\documentclass[onefignum,onetabnum]{siamart190516}
\usepackage{bm}
\usepackage{mathrsfs,amsfonts}
\usepackage{stmaryrd}
\usepackage{xcolor}
\usepackage{subcaption}

\usepackage{float,graphicx}
\graphicspath{{./figures/}}

\newcommand{\bv}{\bm{v}}
\renewcommand{\i}{\bm{i}}
\newcommand{\bb}[1]{\mathbb{#1}}
\newcommand{\bin}[2]{\begin{pmatrix}#1\\ #2\end{pmatrix}}
\newcommand{\pd}[2]{\frac{\partial #1}{\partial #2}}

\begin{document}
\title{Arbitrary order principal directions and how to compute them}

\author{Julie Digne\thanks{LIRIS, CNRS, Univ Lyon
  (\email{julie.digne]liris.cnrs.fr})}
\and Sébastien Valette \thanks{CREATIS, CNRS, Univ Lyon  (\email{sebastien.valette@creatis.insa-lyon1.fr})} 
\and Rapha\"elle Chaine \thanks{LIRIS, CNRS, Univ Lyon  (\email{raphaelle.chaine@liris.cnrs.fr})} 
\and Yohann B\'earzi \thanks{LIRIS, CNRS, Univ Lyon  (\email{yohann.bearz@liris.cnrs.fr})} 
}
\maketitle

\begin{abstract}
Curvature principal directions on geometric surfaces are a ubiquitous concept of Geometry Processing techniques. However they only account for order $2$ differential quantities, oblivious of higher order differential behaviors. In this paper, we extend the concept of principal directions to higher orders for surfaces in $\bb R^3$ by considering symmetric differential tensors. We further show how they can be explicitly approximated on point set surfaces and that they convey valuable geometric information, that can help the analysis of 3D surfaces.
\end{abstract}


\begin{figure}
\begin{center}
\includegraphics[width=0.3\textwidth]{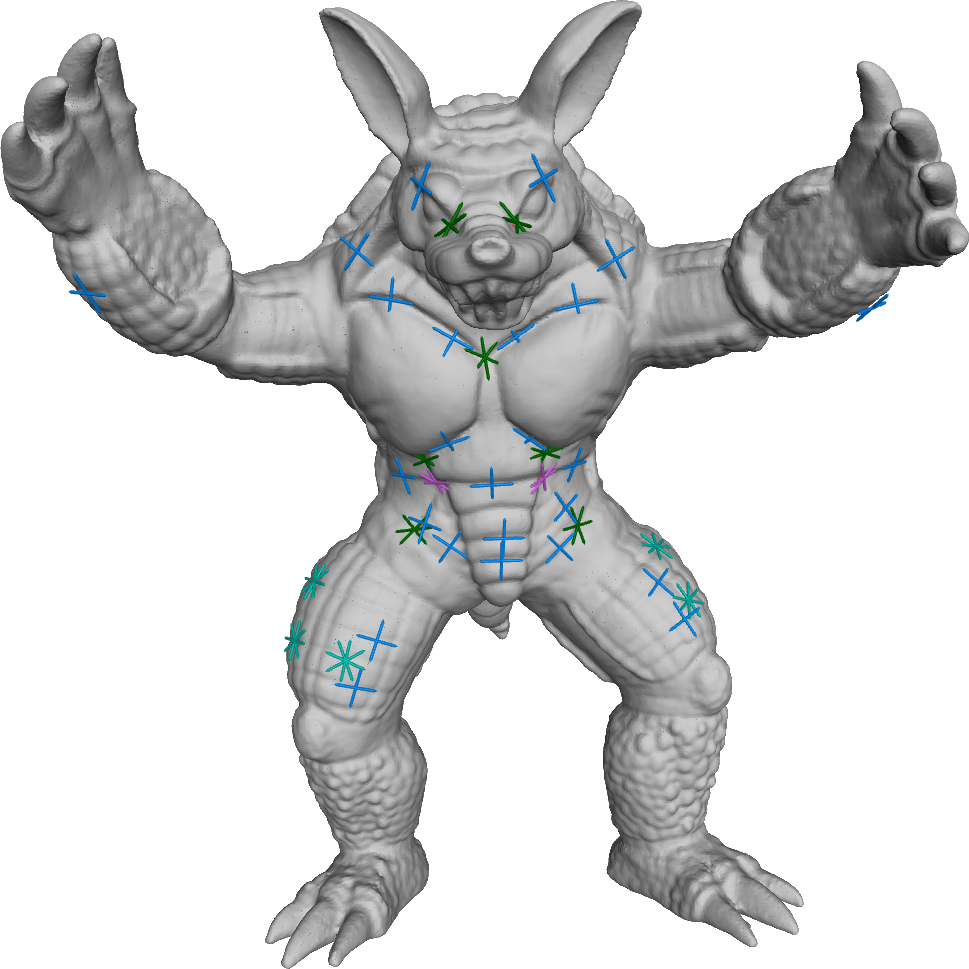}
\includegraphics[width=0.22\textwidth]{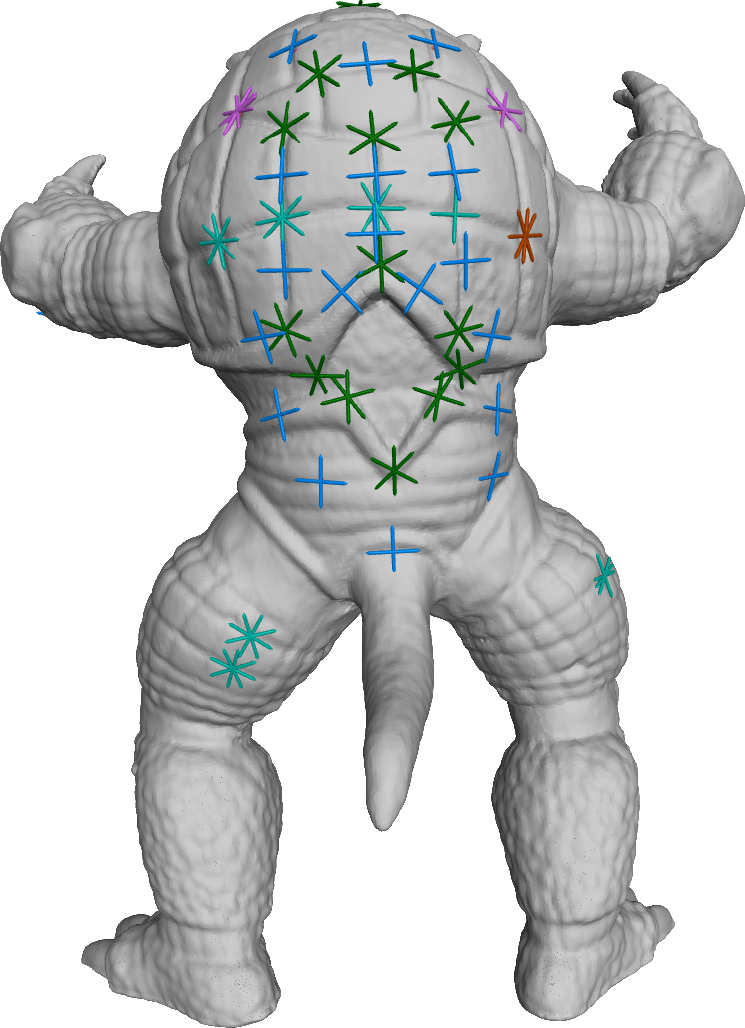}
\includegraphics[width=0.33\textwidth]{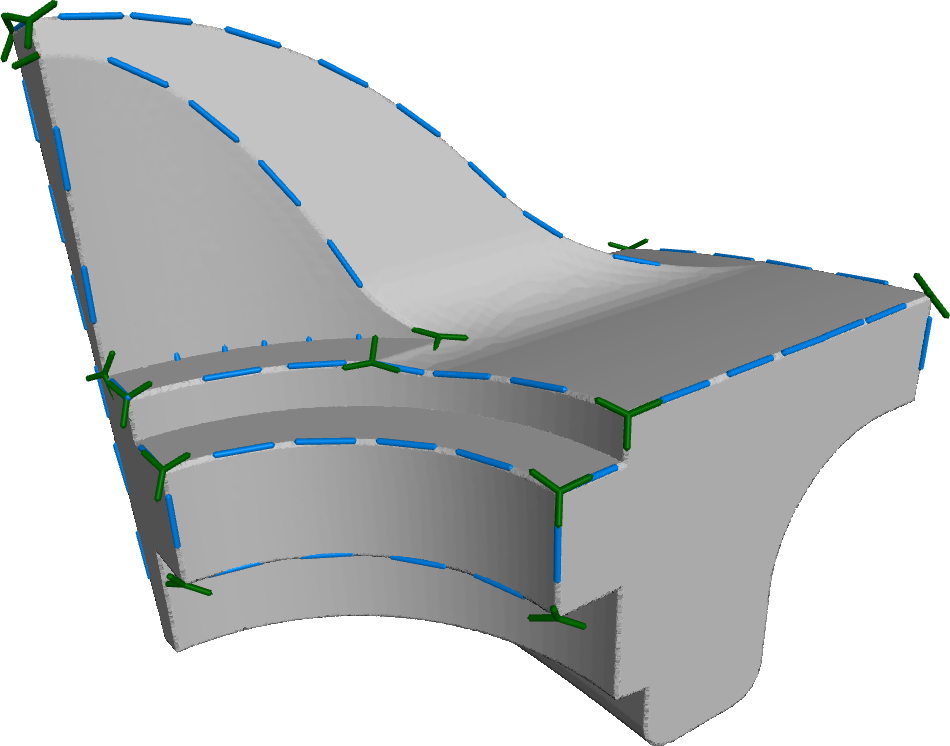}
\end{center}
\caption{Examples of principal directions of arbitrary orders on the Armadillo and Fandisk point sets. Blue: order 2; green: order 3; cyan: order 4; pink: order 5; brown: order 6.\label{fig:teaser}}
\end{figure}

\section{Introduction}

The computation of informative tangent vector quantities on surfaces is a widely studied topic. The most standard vector quantities one can consider are the principal directions that can be estimated via differential geometry tools. However these do not necessarily serve all Computer Graphics purposes: in case of umbilical surface parts, or at monkey saddles, this vector field becomes locally irrelevant.
Further, it only accounts for an \emph{edge} like structure of the shape which is limited.
Our approach considers per point vector quantity estimation, adopting a differential analysis point of view. We strive to go beyond the differential order 2 usually considered when analysing surfaces and extend the definition of principal directions to higher orders. We show experimentally why this definition, beyond its theoretical interest is  appealing for surface analysis.

To summarize, our contributions are as follows:
\begin{itemize}
 \item The mathematical definition of arbitrary order principal directions and their link with symmetric tensor eigenvalues.
 \item A theoretical analysis of their properties.
 \item A practical way of computing the directions on a sampled surface.
\end{itemize}

\section{Related Work}
In this section we review recent works on tangent vector quantities that can be set or estimated on a meshed or sampled surface be they guided by differential properties or designed by a global optimization process with user-prescribed constraints. 

\paragraph*{Differential quantities estimation}

Estimating differential quantities on surfaces has been at the core of Geometry Processing Research. However, surface analysis restricts very often to order 1 and 2 differential properties and has seldom tackled higher order properties.
Among order 2 quantities, the most famous one may be the Laplace-Beltrami operator, whose design has gathered a lot of works both from a theoretical analysis (e.g. \cite{Meyer02,Wardetzky07}) and practical analysis through the Manifold Harmonics Basis (e.g. \cite{Vallet08}). Related to the Laplace-Beltrami operator are the principal curvatures and curvature directions (or equivalently the curvature tensor) estimation, either on a point set surface~\cite{Kalogerakis07,Kalogerakis09} or on meshes~\cite{Cohen-Steiner03}, with applications to curvature lines tracking among others.

It is however possible to get access to high order derivatives of the local surface using local regression, in the context of Moving Least Squares~\cite{Levin98}. Among those methods, Osculating Jets~\cite{Cazals05} express the surface locally as a truncated Taylor expansion wrt to a local planar parametrization. The coefficients can be estimated through a linear system solve and give then a direct access to high order cross derivatives. Interestingly, this approach proved that the error on order $k$ differential quantities estimation in a neighborhood of radius $r$ using a Taylor expansion of order $K$ was in $O(r^{K-k})$. In other words, and quite counter-intuitively, to increase the accuracy of an order 2 estimation, one should still consider a large Taylor expansion order. Several other basis have been proposed following this trend, including the Wavejets basis~\cite{Bearzi18} which is less sensitive to the choice of the local reference frame in the parametrization plane. When the surface is described as a point set, the regression relies on Iteratively Reweighted Least Squares in the presence of noise and/or outliers.
Going further than order 2 Rusinkiewicz~\cite{Rusinkiewicz04} introduces a way to compute curvatures and curvature derivatives but does not go beyond this order.

All these methods essentially perform per point estimation and do not take into account any global regularization constraints. For example, on planar or spherical surfaces curvature directions are erratic in the absence of smoothness constraints, which is required by many computer graphics applications. Hence researchers have turned to the design of consistent vector fields more suited to some application purposes.

\paragraph*{Tangent Vector Field Design}
The problem of tangent vector field design is to compute a smooth tangent vector field with user prescribed constraints at given points of the surface, while optimizing some regularity criterion.
We only review some of the seminal papers of this field and we refer the reader to \cite{Vaxman16} for an extensive review. Many vector field design methods focus on smooth $N$-symmetric vector fields, also known as rotationally-symmetric direction fields (N-RoSy). N-symmetry directions are defined as the sets of directions invariant by $2\pi/N$ rotations. Ray et al.~\cite{Ray08} generalize the notion of singularity and singularity index to these fields, and provide a way of controlling singularities on surface meshes. Lai et al.~\cite{Lai10} focus on casting the vector field design as a Riemannian metric design problem. Further smoothness constraints~\cite{Crane10,Knoppel13} and global symmetry enforcing constraints~\cite{Panozzo12} were proposed. N-RoSy were also extended to non necessarily orthogonal nor rotationally symmetric vector fields \cite{Diamanti14} with appropriate differential operators \cite{Diamanti15} and application to Chebyshev nets computation \cite{Sageman-Furnas19}. 

The generalization of the Laplace-Beltrami operator to tangential vector fields and the subspace raised by its eigenvectors up to a given order may be used for regularizing vector field design \cite{Brandt17}.
Following this development of subspace methods for tangent field design~\cite{Brandt18}, Nasikun et al.~\cite{Nasikun20} consider tangent vector field design and processing via locally-supported tangential fields leading to fast approximation and design algorithms.

\paragraph*{Application of Tangent Vector Fields}
Applications of tangential vector fields range from texture mapping and texture synthesis on surfaces (e.g. \cite{Wei01,Knoppel15}) to fluid simulation (e.g. \cite{Azencot15}). Shape reconstruction and quad-meshing have been tackled by combining a position field and a N-RoSy \cite{Jakob15} yielding an extremely fast interactive algorithm.

In this paper we are also interested in computing per point sets of directions that are neither necessarily orthogonal nor rotationally symmetric but these directions stem from arbitrary order differential properties of the surface. Hence smoothness is obtained by continuity of the underlying mathematical surface.

\section{Arbitrary order symmetric tensor}

Our work makes extensive use of symmetric tensors and the theory developed by Qi for their spectral analysis~\cite{Qi05,Qi06,Qi07}. 

\begin{definition}
An $m$-dimensional symmetric tensor $T$ of order $k$ is a $k$-dimensional array such that given a set of indices $I=(i_j)_{j\in\llbracket0,k\rrbracket}$ with $1 \leq i_j \leq m$, for any permutation $p$ on $I$, $T_I=T_{p(I)}$ 
\label{def:symmetry}
\end{definition}

Notice that in Qi's work, a distinction is made between a tensor and a supermatrix, that is a tensor's representation in a given basis. For clarity sake, here we rather  use the tensor term for both the object and its representation in a basis.

From now on, we will always consider $m=2$ since we are interested in tensors of differential properties related to surfaces of dimension 2. In this case, a symmetric tensor of order $k$ can be seen as a $k-1$-dimensional array of vectors of length 2.
By convention, we define any vector to be a symmetric tensor of order 1.
Given a vector $\bm{v} = (x,y)^T$, we note $\bm{v}^k$ the symmetric tensor of order $k$ generated by multiplying $\bm{v}$ $k$ times using the Cartesian product. In particular, we set $\bm{v}^0=1$ by convention.

Multiplying a symmetric tensor by a vector $\bm{v}$ produces a symmetric tensor of order lowered by $1$. Let $T$ be a symmetric tensor of order $k$, it is composed of two symmetric tensors of order $k-1$ $T_{\bm{x}}$ and  $T_{\bm{y}}$ and can be written $T=(T_{\bm{x}}, T_{\bm{y}})$. Then $T\bm{v}$ is the symmetric tensor of order $k-1$ such that :
\begin{equation}
T\bm{v} = x T_{\bm{x}} + y T_{\bm{y}}
\end{equation}

The product $T\bm{v}$ reduces the order of $T$ by contracting on an arbitrary index. Since $T$ is symmetric, any index used for the contraction yields the same list of numbers in $T\bm{v}$.
While eigendecomposition of matrices is a well understood theory with important applications in Geometry Processing among others, the generalization to arbitrary order tensors is highly nontrivial. Qi introduced a new definition extending eigenvalues and eigenvectors to symmetric tensors~\cite{Qi05,Qi06,Qi07}, as follows:

\begin{definition}{E-eigenvalues \cite{Qi05}}
 Given $T$ a symmetric tensor of order $k$, if there exists
$\lambda\in\mathbb{C}$ and a vector $\bm{v}\in\mathbb R^2$ such that:
\begin{equation}
\left\{
\begin{array}{ccc}
T \bm{v}^{k-1}& =& \lambda \bm{v}\\
\bm{v}^T\bm{v} &=& 1
\end{array}
\right.
\label{eq:eigen_tensor}
\end{equation}
Then $\lambda$ is called an $E$-eigenvalue of $T$ and $\bm{v}$ is called an $E$-eigenvector of $T$.
The set of $\lambda$ satisfying (\ref{eq:eigen_tensor}) are the roots of a polynomial called the $E$-characteristic polynomial.
\end{definition}

\section{Arbitrary order principal directions}

\paragraph{Differential tensor}
Tensors can be used to write Taylor expansions. As an example, one can write the two first terms of a bivariate Taylor expansion. Given $\bm{v} = (x,y)^T$ an arbitrary vector, $\bm{n}$ the normal at $(0,0)$ and $H$ the Hessian of a function defined on $\bb R^2$ with values in $\bb R$ and twice continuously differentiable at $(0,0)$:
\begin{equation}
f(\bm{v}) = f(0,0) + \bm{n}^T\bm{v} + \frac12\bm{v}^TH\bm{v} +o(\|\bv\|^2)
\end{equation}

Note that $H$ is symmetric, and so is $\bm{n}$ since its order is 1. This expression can be generalized to higher orders Taylor expansions using tensors.

The Taylor expansion of order $K$ of a function $f$ from $\bb R^2$ to $\bb R$ is:
\begin{equation}
 f(x,y) = \sum_{k=0}^K\sum_{i=0}^k \frac1{i!(k-i)!}\frac{\partial^k f }{\partial x^i\partial y^{k-i}}(0,0) x^i y^{k-i} +o(\|(x,y)\|^K)
 \label{eq:taylor}
\end{equation}

The differential tensor is now defined as:
\begin{definition}
For a function $f$ defined on $\bb R^2$ with values in $\bb R$, $k$ times differentiable, the $k^{th}$ order differential tensor $T_k$ at $(0,0)$ is a $k^{th}$ order 2-dimensional tensor whose coefficients are:
\begin{equation}
(T_k)_{i_1\cdots i_k} = \frac{\partial^k f}{\partial x_{i_1}\cdots \partial x_{i_k}}(0,0)  
\end{equation}
with $i_j\in \{1,2\}$, for $j\in\{1\cdots k\}$, and $x_1=x,x_2=y$
\label{def:difftensor}
\end{definition}

If $f$ is differentiable then the order in which the differentiation is done is irrelevant and thus $T_k$ is symmetric.
Using Definition \ref{def:difftensor} and $\bv=(x,y)^T$, we have:

\begin{equation}
T_k\bv^k = \sum_{i=0}^k\bin{k}{i} \frac{\partial^k f}{\partial^i x\partial^{k-i}y}(0,0)x^iy^{k-i}
\end{equation}

and:

\begin{equation}
 \frac1{k!}T_k\bv^k = \sum_{i=0}^k\frac{1}{i!(k-i)!} \frac{\partial^k f}{\partial^i x\partial^{k-i}y}(0,0)x^iy^{k-i}
 \label{eq:tensoreq}
\end{equation}

Hence using Equation \ref{eq:taylor}, we get the Taylor formula for a $K$ times differentiable function:
\begin{equation}
 f(\bv) = \sum_{k=0}^K \frac1{k!}T_k\bv^k +o(\|\bv\|^K)
 \label{eq:taylorT}
\end{equation}

The following Lemma shows that the gradient of each of the terms involved in the Taylor expansion can be obtained by contracting the corresponding tensor $k-1$ times, i.e. one time less than in the  expansion. This will then allow us to search for extrema at different orders.
\begin{lemma}
Let $T$ be a $2$-dimensional symmetric tensor. Let $\bm{v}=(x,y)^T\in\mathbb R^2$ be a vector.
\begin{equation}
\begin{split}
\frac{\partial T\bm{v}^k}{\partial \bm{v}}  &= k T\bm{v}^{k-1}\\
\end{split}
\end{equation}
\label{lemma:deriv_tensor}
\end{lemma}

\begin{proof}
For $k=1$, $\bm{v}^1=\bm{v}=(x, y)^T$ and:
\begin{align}
\frac{\partial T\bv^1}{\partial \bv} = (\frac{\partial xT_x+yT_y}{\partial x}, \frac{\partial xT_x+yT_y}{\partial y}) = (T_x,T_y) = T  
\end{align}

Assume that for $k$, we have $\frac{\partial T\bm{v}^k}{\partial \bm{v}} = k T\bm{v}^{k-1}$, then:
\begin{equation}
\begin{split}
\frac{\partial T\bm{v}^{k+1}}{\partial\bm{v}} &= \frac{\partial T\bm{v}^k\bm{v}}{\partial\bm{v}}\\
&= T\bm{v}^k + \frac{\partial T\bm{v}^k}{\partial\bm{v}}\bm{v}\\
&= T\bm{v}^k + kT\bm{v}^{k-1}\bm{v}\\
&= (k+1)T\bm{v}^k\\
\end{split}
\end{equation}

By induction, the property is true for all $k\geq 1$.

\end{proof}

\begin{theorem}
Given a $K$ times continuously differentiable function $f$ and $1< k \leq K$, $T_k$ is the real symmetric $k^{th}$ order differential tensor of $f$ and the set of vectors $\bm{v}=(r\cos\theta,r\sin\theta)^T$ such that $\frac{\partial}{\partial\theta} T_k\bm{v}^k=0$ and $\|\bm{v}\|=1$ are real $E$-eigenvectors of $T_k$, i.e.:
\begin{equation}
\left\{
\begin{array}{l}
T_k\bm{v}^{k-1} = \bm{v}T_k\bm{v}^k\\
\|\bm{v}\| = 1
\end{array}
\right.
\end{equation}
\label{th:principal_directions}
\end{theorem}

\begin{proof}
First one can notice using equation \ref{eq:tensoreq}, with $\bv=(r\cos\theta,r\sin\theta)^T$ that:

\begin{equation}
 \frac1{k!}T_k\bv = \sum_{i=0}^k a_{k,i} r^k \cos^i\theta\sin^{k-i}\theta
 \label{eq:tensorcartesian}
\end{equation}

where $a_{k,i}=\frac1{i!(k-i)!}\frac{\partial^k f}{\partial x^i\partial y^{k-i}}(0,0)$. 

Differentiating $T_k\bv^k$ w.r.t. radius $r$ gives:

\begin{equation}
\begin{split}
\frac{\partial}{\partial r}T_k\bm{v}^k &= \frac{\partial}{\partial r} k!\sum_{i=0}^k a_{k,i}r^k\cos^i\theta\sin^{k-i}\theta\\
				&= k(k!)\sum_{i=0}^k a_{k,i}r^{k-1}\cos^i\theta\sin^{k-i}\theta\\
				&= \frac{k}{r}k!\sum_{i=0}^k a_{k,i}r^k\cos^i\theta\sin^{k-i}\theta\\
				&= \frac{k}{r}T_k\bm{v}^k
\end{split}
\end{equation}

Since $r\frac{\partial}{\partial r} = x \frac{\partial}{\partial x} + y \frac{\partial}{\partial y}$ we get:

 $$x\pd{T_k\bv^k}{x} + y\pd{T_k\bv^k}{y} = r\pd{T_k\bv^k}{r} =k T_k\bv^k$$

Differentiating w.r.t angle $\theta$ to look for extrema:

$$\pd{T_k\bv^k}{\theta} = 0 \Leftrightarrow  -y \pd{T_k\bv^k}{x} + x\pd{T_k\bv^k}y = 0$$ 

which yields the following equations:

\begin{equation}
\begin{split}
    &\left\{
    \begin{array}{ccc}
      -y \pd{T_k\bv^k}{x} + x\pd{T_k\bv^k}y & = & 0 \\
    x\pd{T_k\bv^k}{x} + y\pd{T_k\bv^k}{y} & = & k T_k\bv^k\\  
    \end{array}\right.\\
\Leftrightarrow
&\left\{
    \begin{array}{ccc}
      -y^2 \pd{T_k\bv^k}{x} - x^2\pd{T_k\bv^k}x & = & -xk T_k\bv^k \\
    x^2\pd{T_k\bv^k}{y} + y^2\pd{T_k\bv^k}{y} & = & y k T_k\bv^k\\
    \end{array}\right.\\
\Leftrightarrow
&\left\{
    \begin{array}{ccc}
      \|\bv\|^2 \pd{T_k\bv^k}{x} & = & xk T_k\bv^k \\
      \|\bv\|^2\pd{T_k\bv^k}{y} & = & y k T_k\bv^k\\
    \end{array}\right.\\
\Leftrightarrow&  \pd{T_k\bv^k}{\bv}  =  \frac{\bv}{\|\bv\|^2} k T_k\bv^k\\
\Leftrightarrow&  k T_k\bv^{k-1}  =  \frac{\bv}{\|\bv\|^2} k T_k\bv^k\\
\Leftrightarrow& T_k\bv^{k-1}  =  \frac{T_k\bv^k}{\|\bv\|^2} \bv\\
    \end{split}
    \label{eq:proof}
\end{equation}

Since we look for real unitary vectors, we add the constraint that $r=\|\bv\|=1$. Moreover, setting $\lambda = T_k\bv^k$, we get $T_k\bv^{k-1}=\lambda \bv$ and $\bv$ is a real E-eigenvalue of $T_k$. The reverse holds using the same equations.
\end{proof}

\begin{definition}
Given a $K$ times continuously differentiable function $f$ defined on $\bb R^2$ with values in $\bb R$, the principal directions of order $k$  ($1< k \leq K$) of $f$ at $(0,0)$ are defined as the real $E$-eigenvectors of its $k^{th}$ order differential tensor $T_k$.
\end{definition}

 One should notice that Qi et al. defined several types of eigenvalues and eigenvectors~\cite{Qi05,Qi06,Qi07}. In particular, if an $E$-eigenvector $\bm{v}$ is real and if its corresponding $E$-eigenvalue $\lambda$ is also real, then $\bm{v}$ is a $Z$-eigenvector and $\lambda$ is a $Z$-eigenvalue. Then our $E$-eigenvalues are also $Z$-eigenvalues in this setting.

Figure \ref{fig:saddles} illustrates the principal directions of order $3$ and $8$ for some illustrative functions at $(0,0)$.

The above form is not very amenable to find the zero-crossings of the derivative of $T_k\bv^k$ with respect to $\theta$. Instead we propose to use its expression in the Wavejets basis (consisting of all functions $(r,\theta)\rightarrow r^k e^{\i n\theta}$ for $k\in \bb N, -k\leq n \leq k$)~\cite{Bearzi18}, hence using polar coordinates:

\begin{equation}
 f(r,\theta) = \sum_{k=0}^K\sum_{n=-k}^k \phi_{k,n} r^k e^{\i n\theta} +o(r^K)
 \label{eq:wavejets}
\end{equation}

with $\phi_{k,n}$ the Wavejets decomposition coefficients. Among other properties, $\phi_{k,n}=\phi_{k,-n}^*$ and $\phi_{k,n}=0$ is $k$ and $n$ do not share the same parity (see \cite{Bearzi18} for more details).

By identification of the coefficients in front of the powers of $r$ with the coefficients of the Taylor expansion using tensors, we have:

\begin{equation}
\frac1{k!}T_k\bv^k = \sum_{n=-k}^k \phi_{k,n} r^k e^{\i n\theta}
 \label{eq:identif}
\end{equation}

\begin{corollary}
Given a function $f$ defined on $\bb R^2$ with values in $\bb R$, $K$ times differentiable at $(0,0)$, the principal directions of order $k$ of $f$ correspond to the real $E$-eigenvectors of $T_k$ and they can be retrieved out of the Wavejets decomposition of $f$ by looking at the zeros of:

\begin{equation}
\frac{\partial}{\partial\theta}\sum_{n=-k}^k \phi_{k,n}e^{\bm{i}n\theta} = \sum_{n=-k}^k
\bm{i}n\phi_{k,n}e^{\bm{i}n\theta}
\label{eq:wavejets_eigen}
\end{equation}
\end{corollary}

\begin{proof}
As shown in Theorem \ref{th:principal_directions}, the $E$-eigenvectors directions correspond to the zeros of the angular derivative of $T_k\bm{v}^k$. Thus, a direct angular differentiation of Equation \ref{eq:identif} yields the result. 
\end{proof}

Since coefficients $\phi_{k,n}=\phi_{k,-n}^*$ in the Wavejet decomposition of a real function, the zero-crossings of Equation \ref{eq:wavejets_eigen} can be obtained by solving the following equation:

\begin{equation}
\sum_{n=1}^k n ( Im(\phi_{k,n})\cos(n\theta)+Re(\phi_{k,n})sin(n\theta)) = 0
\end{equation}
A more convenient form to find roots, for example using Newton's method.

So far, we defined the principal directions as the eigenvectors of a tensor which we linked to the extrema of a function $g_k(\theta)$. The eigenvalues of the tensor can be also linked to this function. Calling $\theta_e$ an angle corresponding to an extremum of $g_k$, the corresponding eigenvalue $\lambda_e$ can be recovered as:

$$g_k(\theta_e) = \frac{\lambda_e}{k!} $$

This follows directly from the last equality in \ref{eq:proof}.

\begin{definition}
Among all principal directions, we call \textit{Maximum} principal directions (resp. \textit{minimum principal directions}) the directions that correspond to local maxima (resp. local minima) of $g_k(\theta) = \sum_{n=-k}^k \phi_{k,n}e^{\bm{i}n\theta}=\frac{T_k\bm{v}^k}{k!r^k}$
with $\bm v = (r \cos \theta, r \sin \theta)$.
\label{def:maxmin}
\end{definition}

\section{Properties of the principal directions}
\label{sec:prop}

\begin{figure}
 \begin{center}
  \includegraphics[height=0.18\textheight]{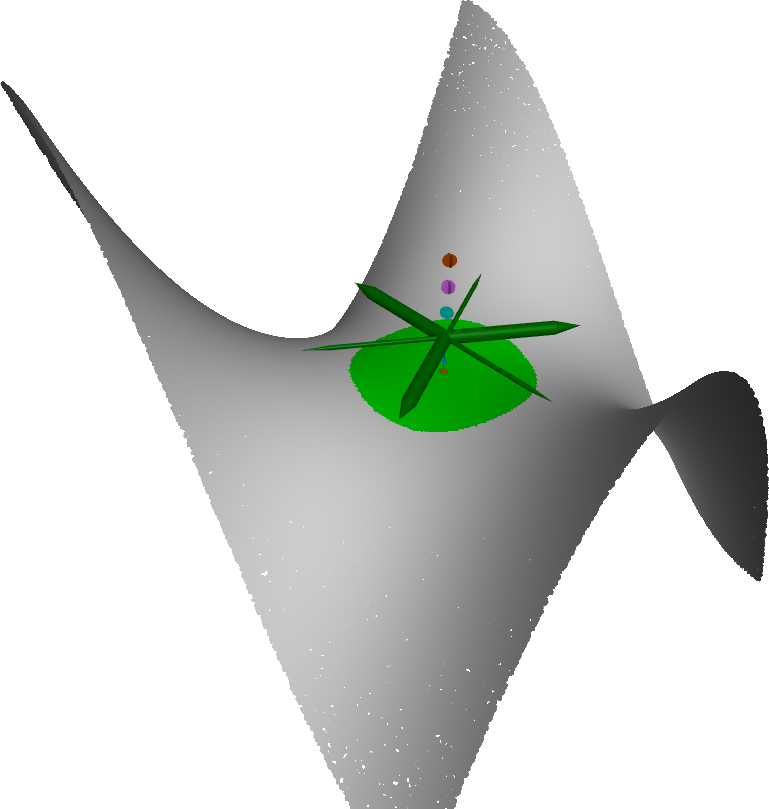}
  \includegraphics[height=0.18\textheight]{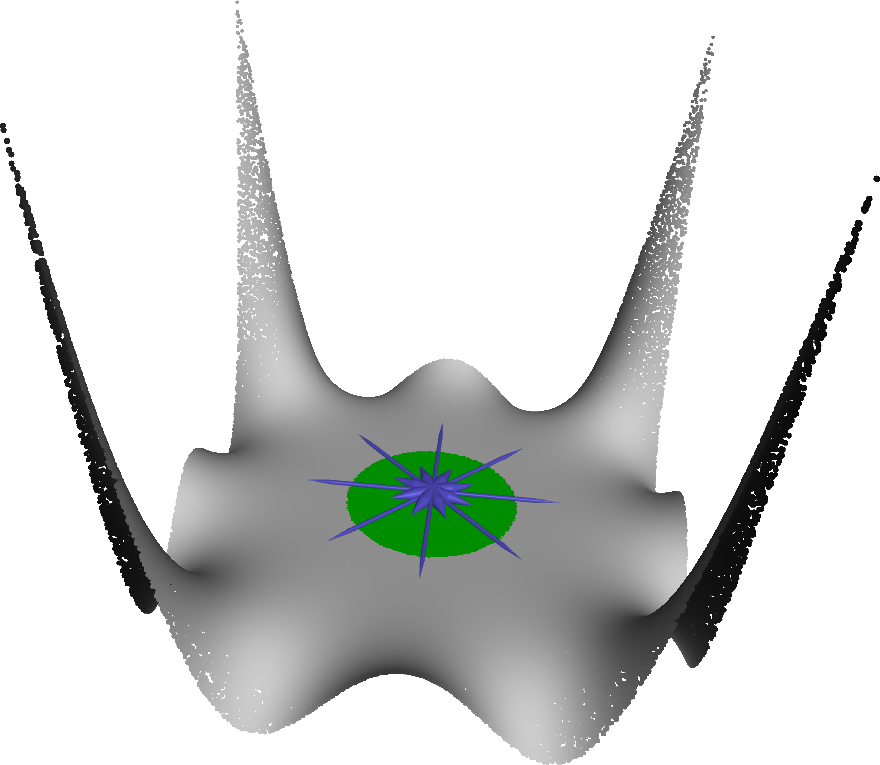}
 \end{center}
\caption{Two synthetic surfaces with relevant principal directions of order $3$ (monkey saddle, left) and order $8$ (octopus saddle, right). Other orders vanish and exhibit no principal directions.}
\label{fig:saddles}
\end{figure}

\paragraph{Constraints on the principal directions} 
The functions $g_k$ (Definition \ref{def:maxmin}) can be rewritten as $g_k(\theta)= \phi_{k,0}+2\sum_{n=1}^k ( Re(\phi_{k,n})\cos(n\theta)-Im(\phi_{k,n})sin(n\theta))$. From the periodicity of cosine and sine functions, we deduce that:
\begin{itemize}
	\item If $k$ is even then $g_k(\theta)=g_k(\theta+\pi)$, hence if $\theta_{0}$ corresponds to a maximum principal direction, $\theta_{0}+\pi$ is also a maximum principal direction.
	\item If $k$ is odd, $g_k(\theta)=-g_k(\theta+\pi)$, hence if $\theta_{0}$ corresponds to a maximum principal direction then $\theta_0+\pi$ corresponds to a minimum principal direction.
\end{itemize}

\paragraph{Number of principal directions} 
There are at most $2k$ principal directions of order $k$, since finding the zeros of $\pd{g_k}{\theta}$ amounts to finding the $1$-norm roots of a real polynomial of order $2k$ 
(obtained by multiplying Equation \ref{eq:wavejets_eigen} by $e^{\bm{i}k\theta}$).
Since two maximum principal directions should be separated by one minimum principal direction and conversely, the number of maximum principal directions and the number of minimum principal directions should be equal and their maximum number is thus $k$.
Following the parity constraints on the location of maxima and minima above, for an even order, the number of maximum principal directions is necessarily even. For similar reasons, for an odd order, the number of maximum principal directions is necessarily odd.

\paragraph{Order 2 principal directions} The principal directions of order $2$ correspond to the classical principal curvature direction, however the maximum (resp. minimum) principal directions might not correspond to the maximum (resp. minimum) curvature directions.

\paragraph{Regularity and link with N-RoSy} Order $k$ principal directions can turn into a $k$-RoSy, if and only if $\phi_{k,n}=0$ for all $n\neq \pm k$. The principal direction distribution can however not be arbitrary: this can be seen by considering order $3$ principal directions, fixing their 3 angles and trying to solve for the coefficients yielding a $0$ derivative for these 3 angles. This amounts to considering $4$ unknowns corresponding to the real and imaginary parts of the coefficients $\phi_{3,1}$ and $\phi_{3,3}$ (since $\phi_{3,n}=\phi_{3,-n}^*$), linked by 6 equations given by $\theta_i$ and $\theta_i+\pi$ with $i=1\cdots 3$.
A rank analysis yields that the system is sometimes invertible and hence yields only the trivial solution of all coefficients set to $0$. Sometimes the system has rank $2$ or $3$ depending on the chosen angles and hence yields a nontrivial solution. Hence not all kind of irregularities can be represented by the principal directions of the tensor.
On Figure \ref{fig:saddles} we show a monkey saddle and an octopus saddle, whose respective order 3 and 8 principal directions correspond to $3$ and $8$-RoSy when considering only maximum (alternatively only minimum) principal directions.

\section{Directions Computation per point}

We now propose to compute these directions on point set surfaces by using a local parametrization around each point. In most Geometry Processing applications, surfaces are known only through a set of points, sometimes connected into a mesh, and the surface in between should be inferred by regression to estimate principal directions.

In practice, to compute principal directions, we perform a local Wavejets surface regression with a high enough order ($K=10$ in our experiments).
More precisely, let $p$ be a point of the point set, let $(p_i)_{i=1\cdots N}$ be its neighbors within some user-defined radius $r$. We compute a local tangent plane by robust PCA and deduce a local parametrization plane on which we choose an arbitrary tangent direction which serves as the origin direction for computing the local polar coordinates $(r_i,\theta_i,z_i)$ for each $p_i$. Then we solve for the Wavejets coefficients $\phi_{k,n}$ by minimizing:

\begin{equation}
 \sum_{i=1}^N w_i\|z_i - \sum_{k=1}^K\sum_{n=-k}^k \phi_{k,n} r_i^k e^{in\theta_i} \|_p^p
 \label{eq:estim}
\end{equation}

with $w_i = \frac1C\exp{-\frac{\|p - p_i\|^2}{2*\sigma^2}}$ and $\sigma = \frac r3$. This Gaussian weight avoids sharp boundary effects and makes the Wavejets estimation smoother in the ambient space.

Depending on the type of data, we can use the $\ell^1$-norm ($p=1$) when there is noise and outliers, or the $\ell^2$ norm ($p=2$) if the data has low noise. As shown by Levin~\cite{Levin98,Levin15}, the coefficients obtained by Moving Least Squares are continuously differentiable if a $\ell^2$ norm is used. The use of the $\ell^1$ norm does not provide such a guarantee. Hence depending on the required smoothness one should use a different norm in the estimation procedure of Equation \ref{eq:estim}.

\section{Experiments}

\paragraph*{Synthetic data}
To illustrate the behavior of arbitrary order principal directions, we compute order $3$ principal directions on a synthetic surface controlled by its Wavejets coefficients (Figure \ref{fig:order3}). The number of maximum directions is either 1 or 3 (hence either $1$ or $3$ minimum principal directions).

\begin{figure*}[h!]
\begin{center}
 \includegraphics[width=0.19\linewidth]{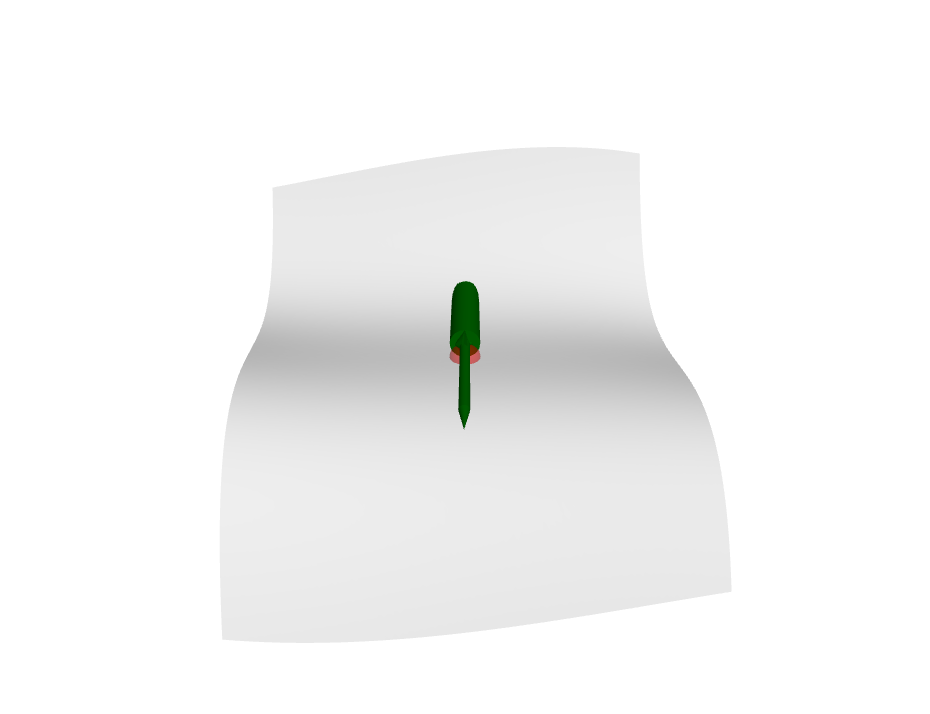}
 \includegraphics[width=0.19\linewidth]{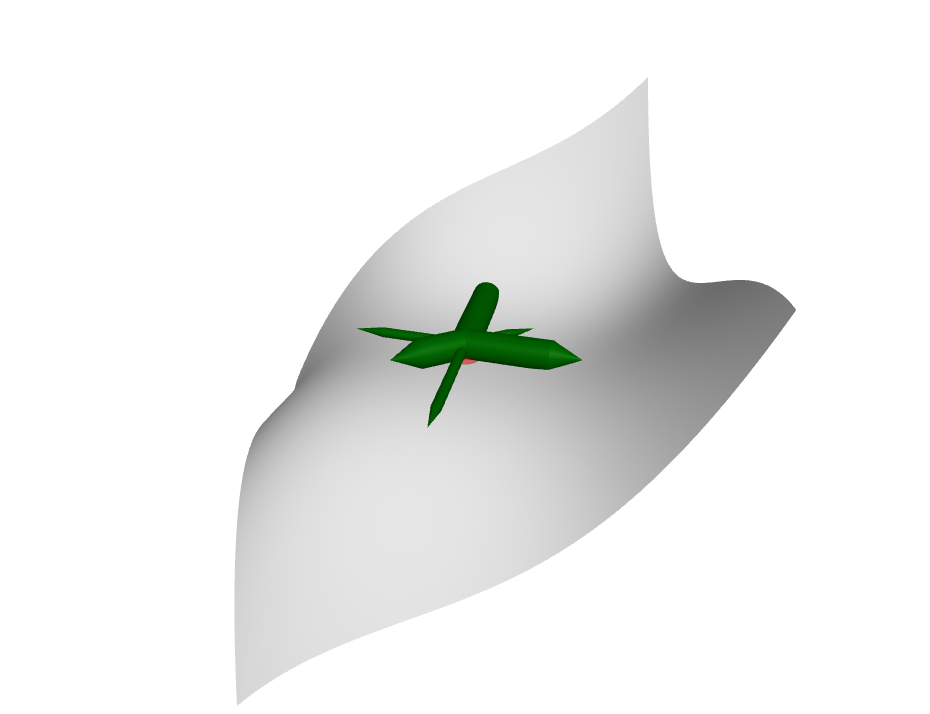}
 \includegraphics[width=0.19\linewidth]{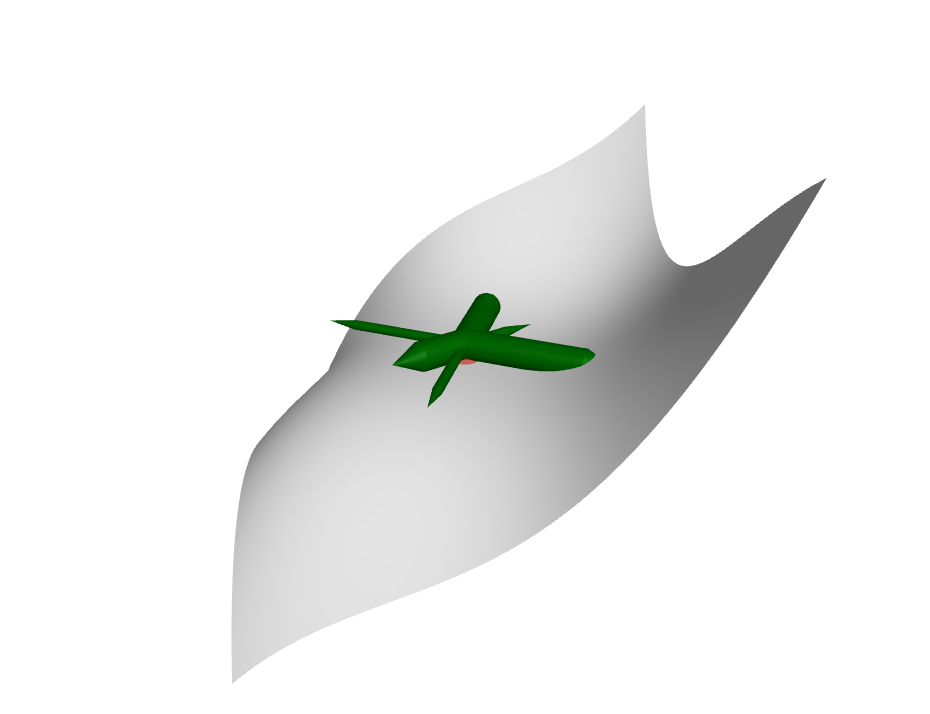}
 \includegraphics[width=0.19\linewidth]{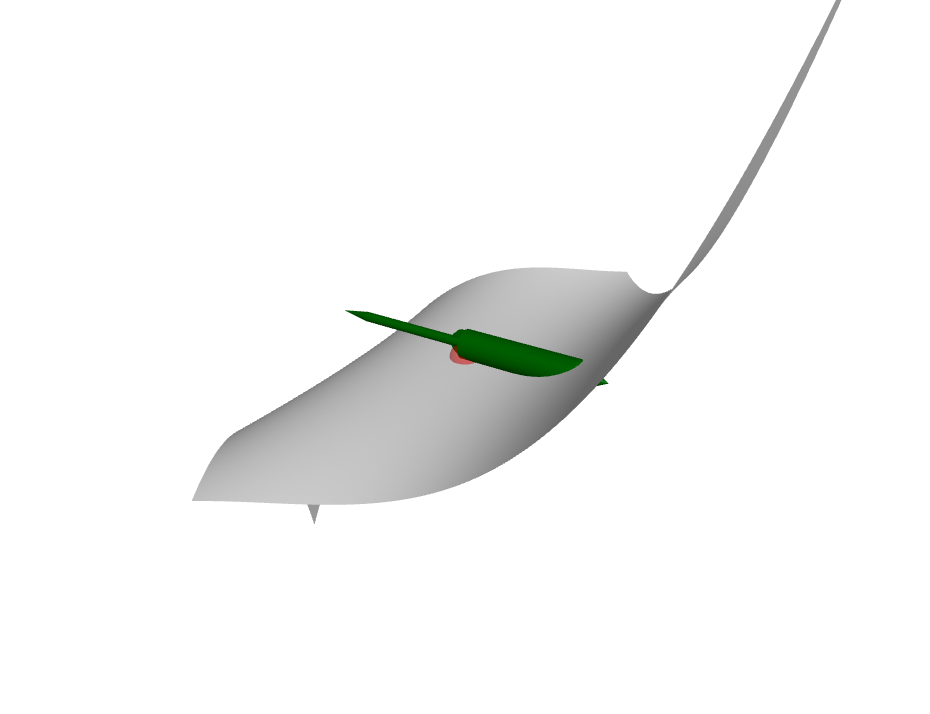}
 \includegraphics[width=0.19\linewidth]{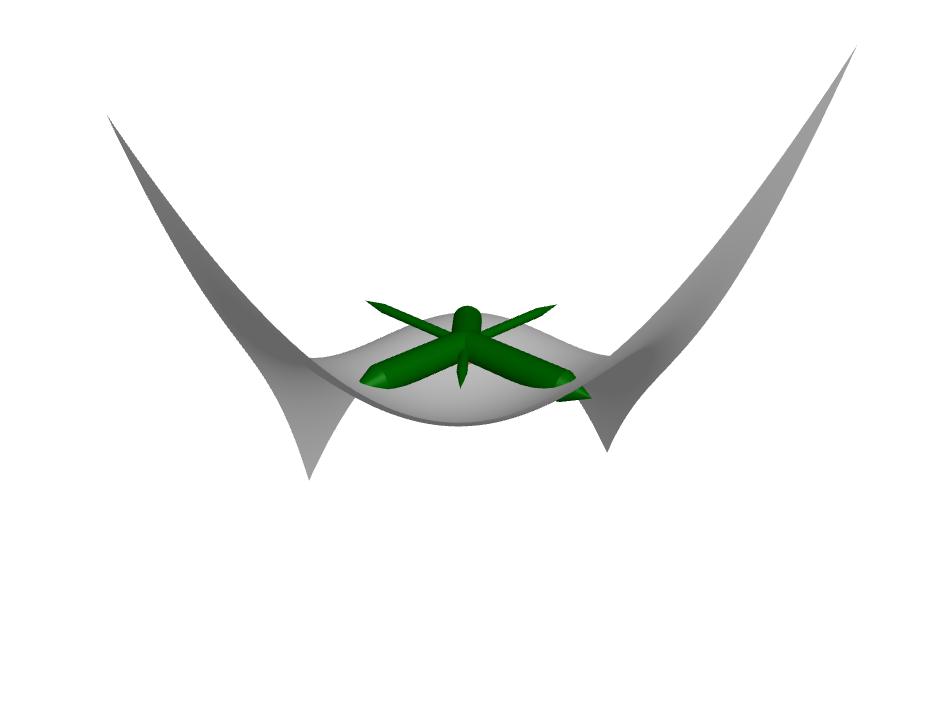}
\end{center}
 \caption{Order $3$ principal directions on a synthetic surface controlled by its Wavejets coefficients.\label{fig:order3}.}
\end{figure*}

On Figure \ref{fig:ridge2junction} we show order 2 and 3 principal directions on a smooth synthetic surface evolving from a ridge to a smooth T-junction. One can see that order $3$ takes slowly over order 2, with a preferred direction.

Figure \ref{fig:inter5plane_irreg} illustrates the behavior of orders 2 to 7 principal directions on a sharp feature created by 5 intersecting planes. No order alone captures all intersection directions, but orders 5 and 7 contain the correct directions. Interestingly, order 7 degenerates to create only 5 maximum principal directions.

\begin{figure}[h!]
 \begin{center}
  \includegraphics[width=0.25\linewidth]{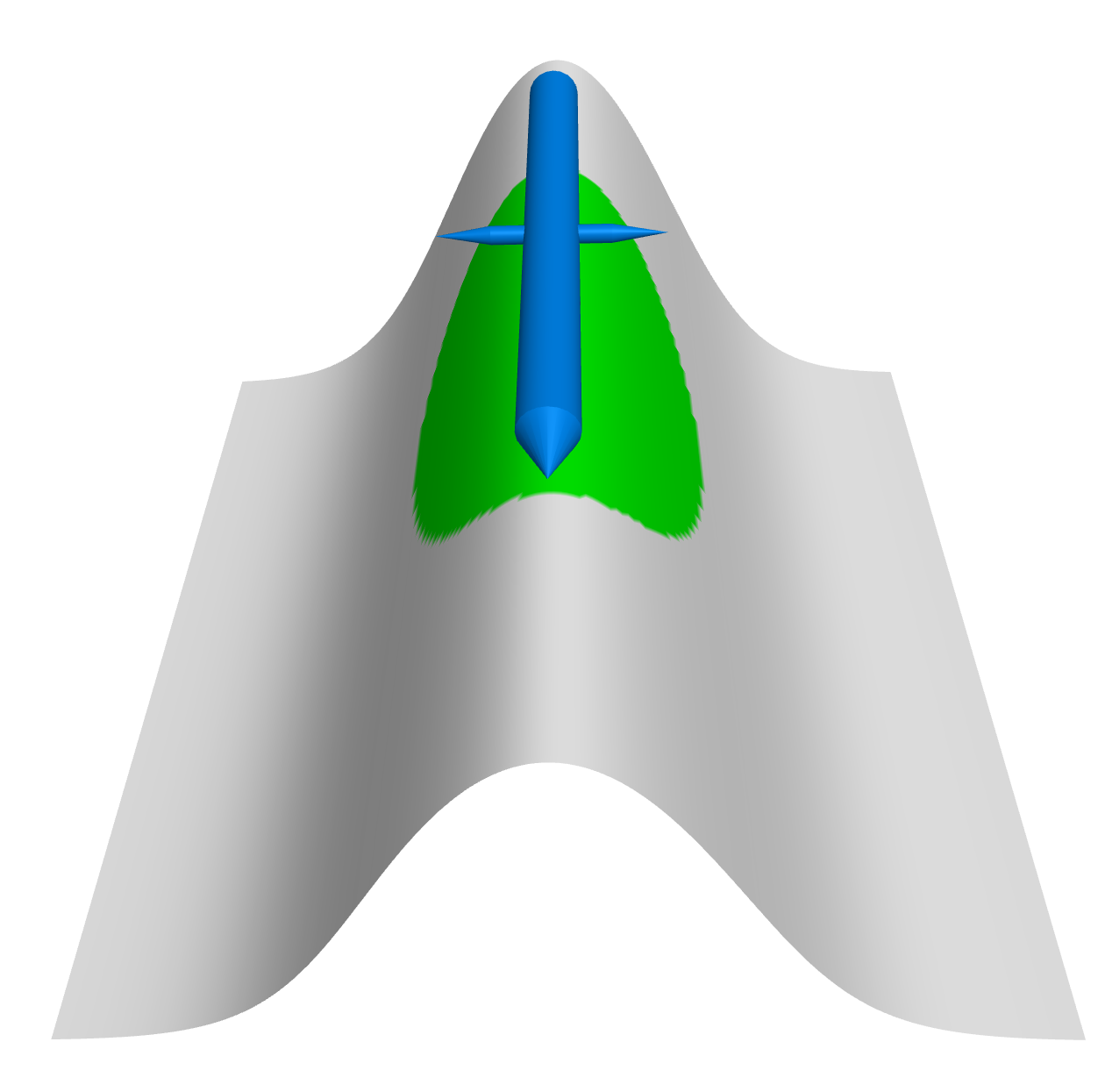}
  \includegraphics[width=0.25\linewidth]{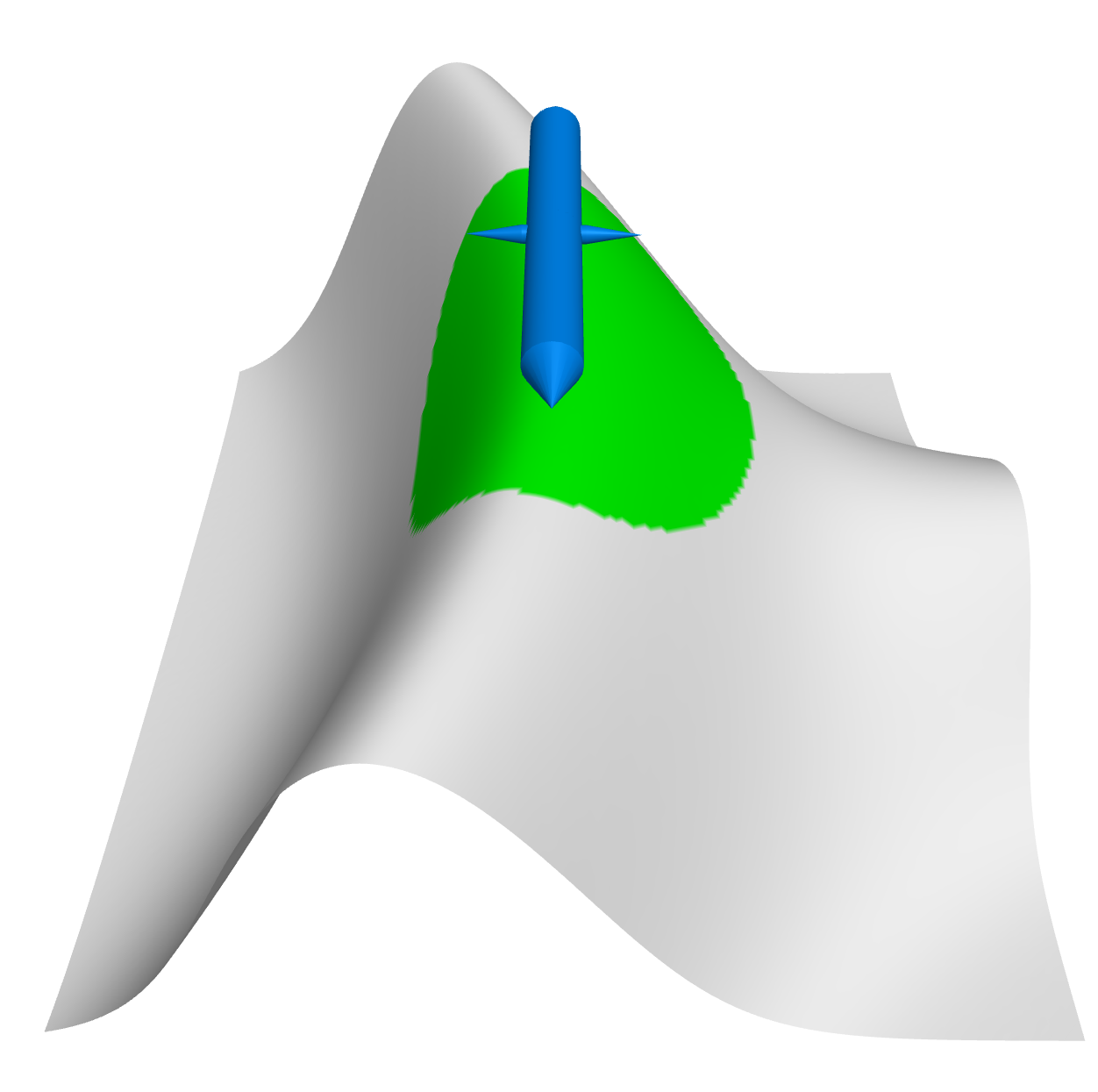}
  \includegraphics[width=0.25\linewidth]{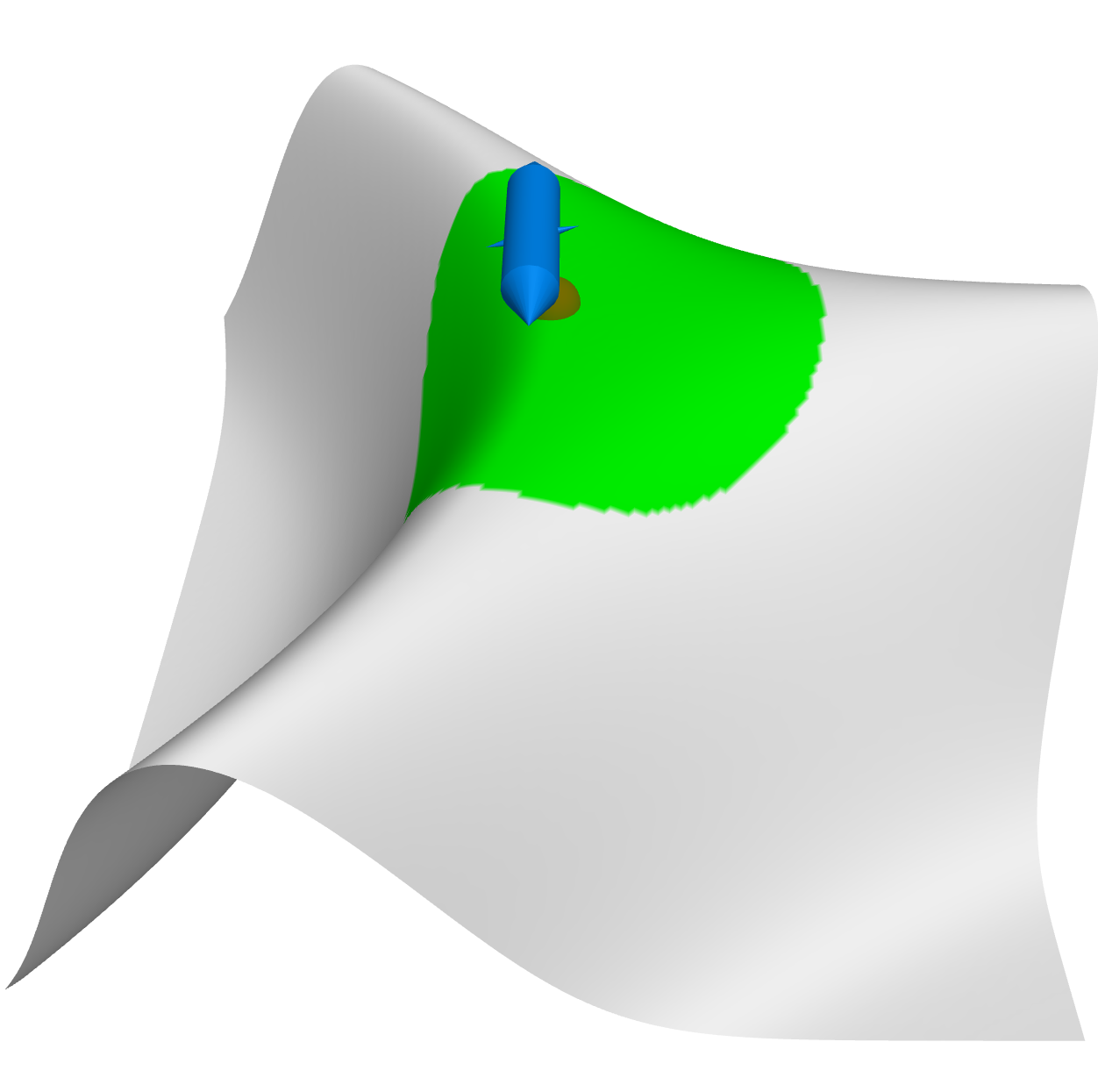}\\
  \includegraphics[width=0.25\linewidth]{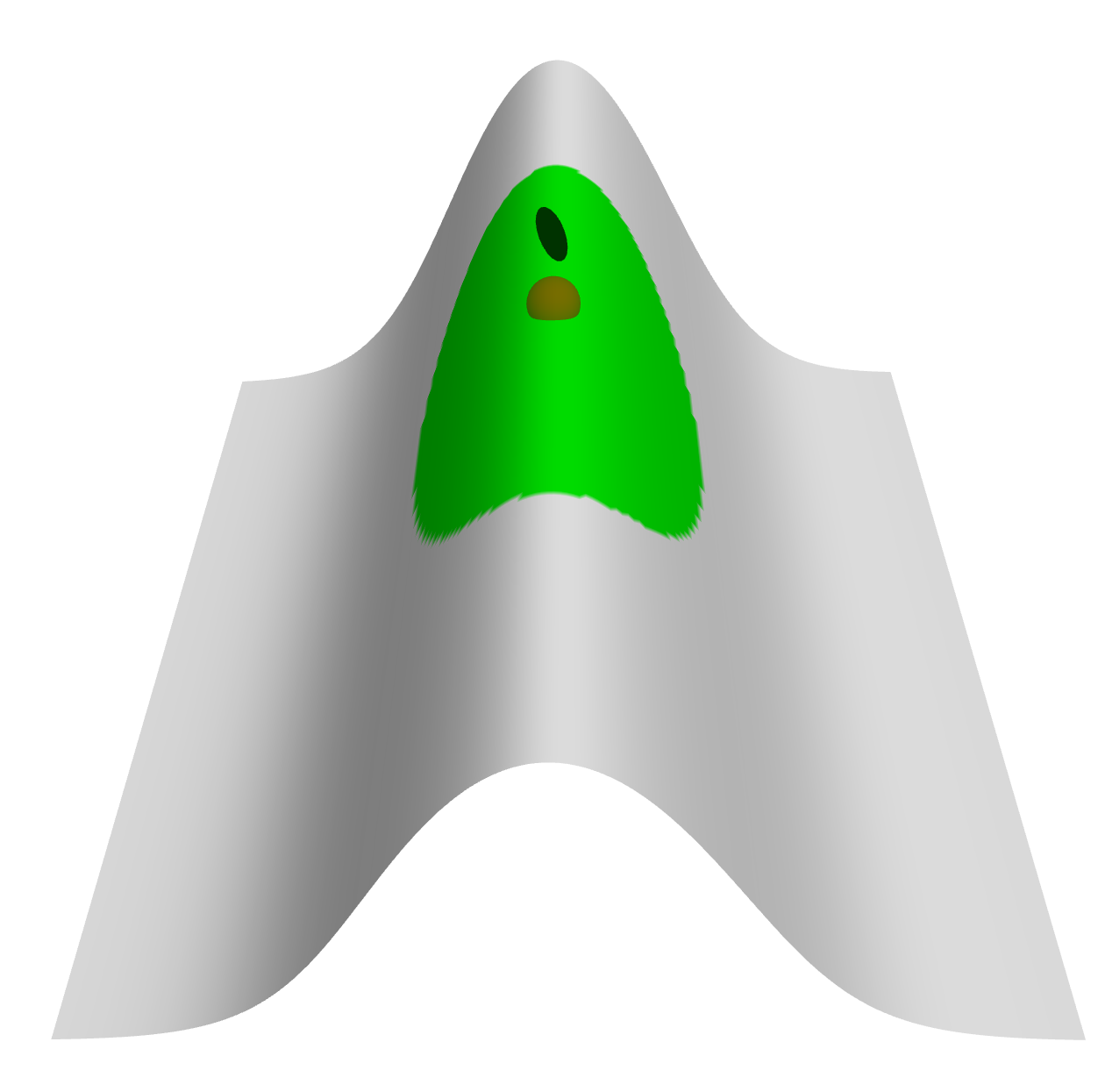}
  \includegraphics[width=0.25\linewidth]{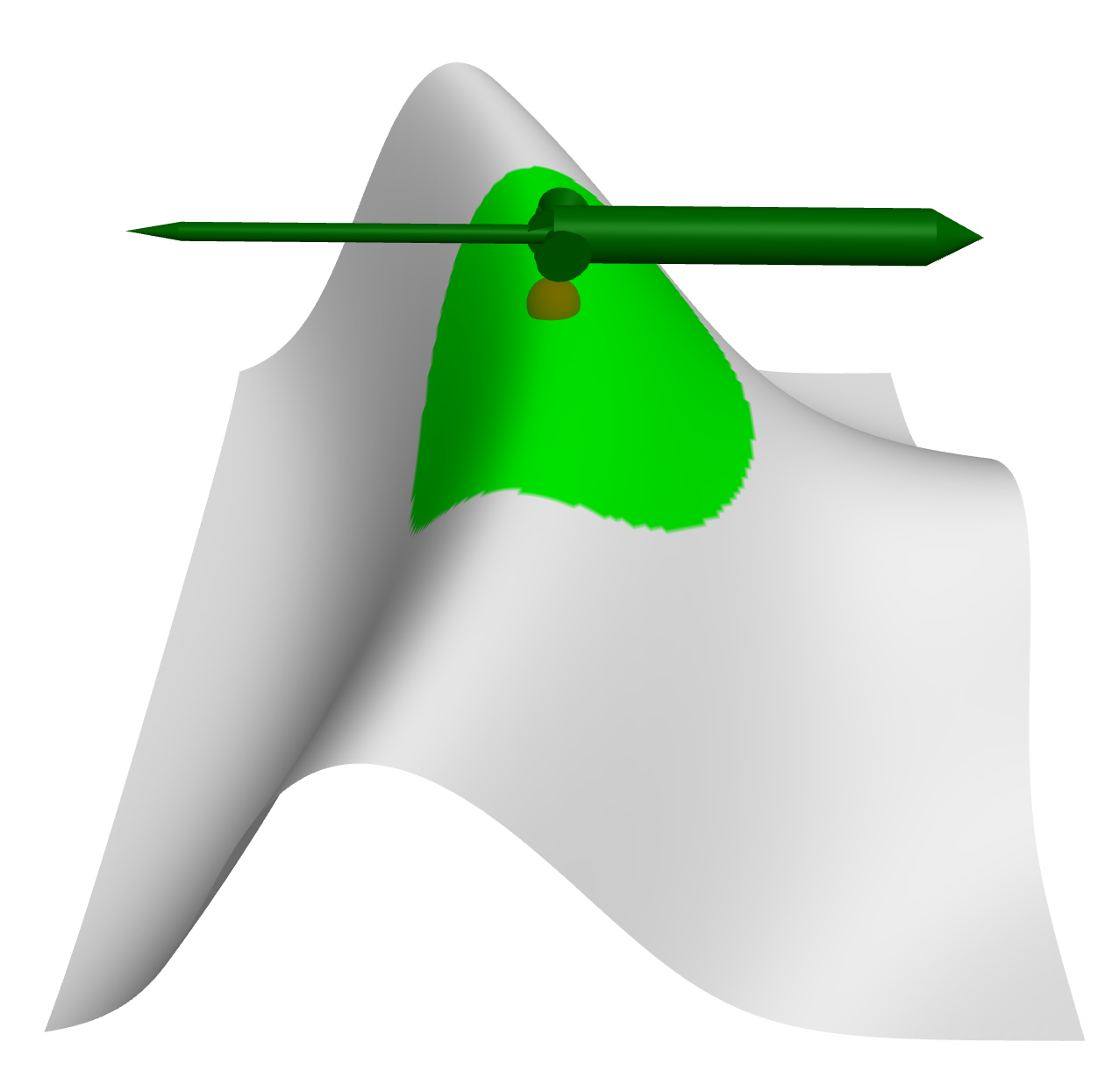}
  \includegraphics[width=0.25\linewidth]{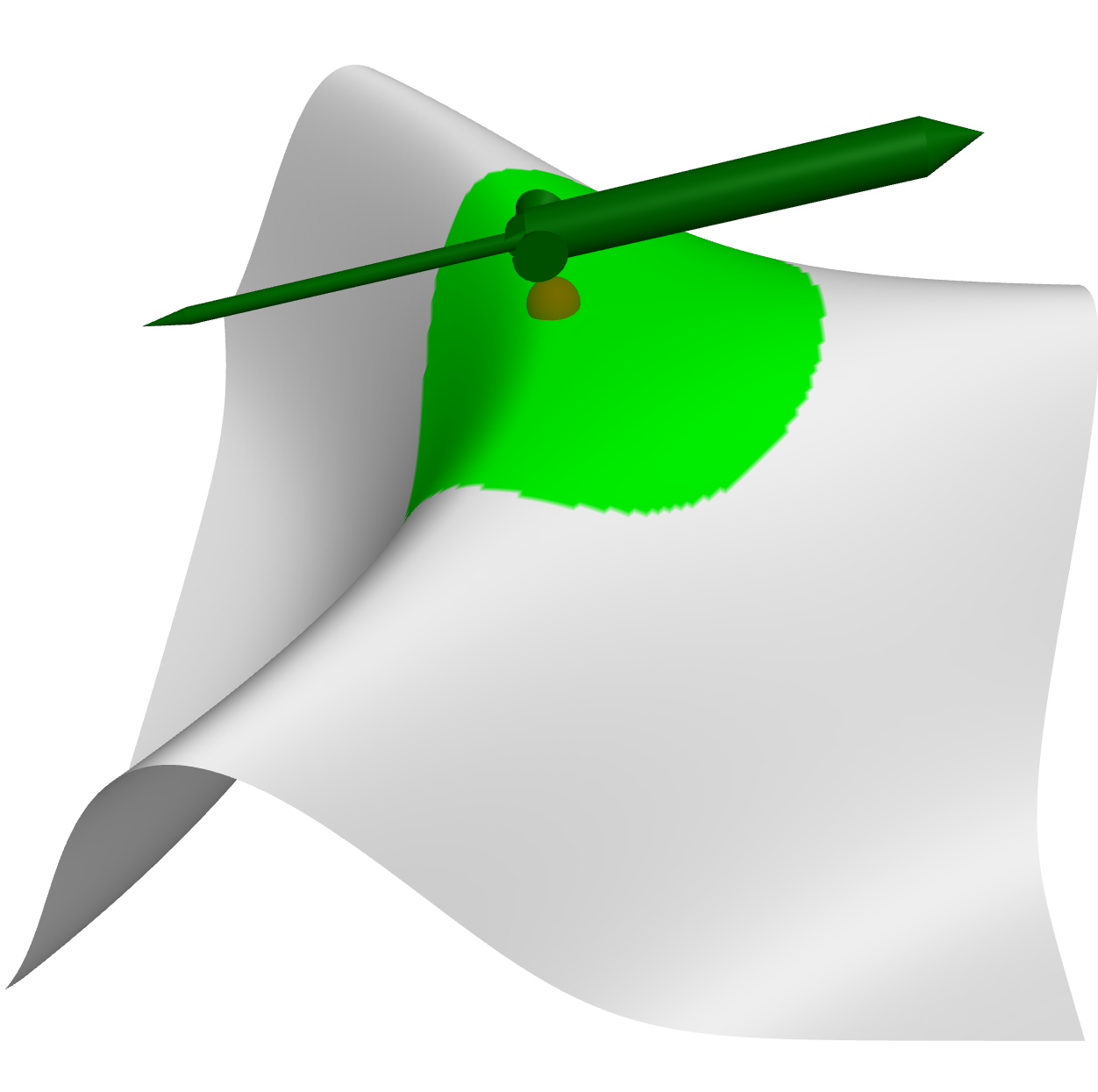}
 \end{center}
\caption{Order 2 (top) and 3 (bottom) principal directions on a surface evolving from a ridge (left) to a smooth T-junction (right). The amplitude of the eigenvector corresponds to the corresponding absolute function value.}
\label{fig:ridge2junction}
\end{figure}

\begin{figure}[h!]
 \begin{center}
  \subcaptionbox{Order 2}{\includegraphics[width=0.15\linewidth]{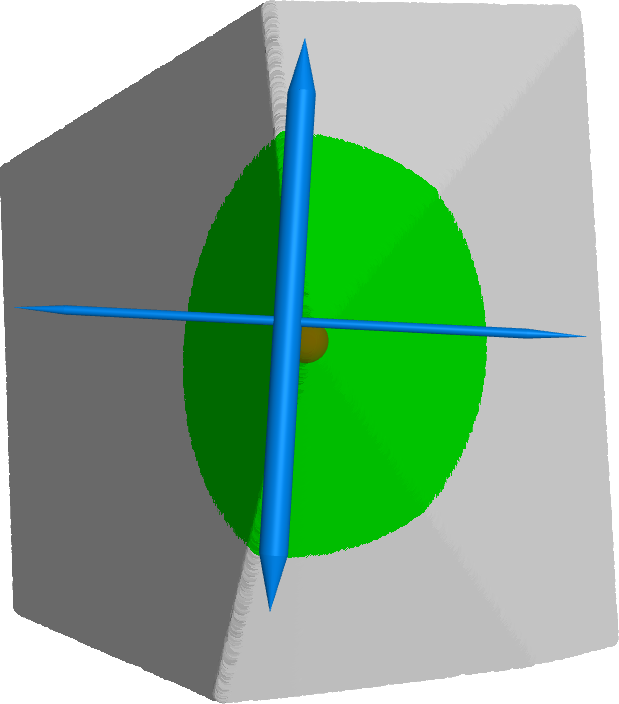}}
  \subcaptionbox{Order 3}{\includegraphics[width=0.15\linewidth]{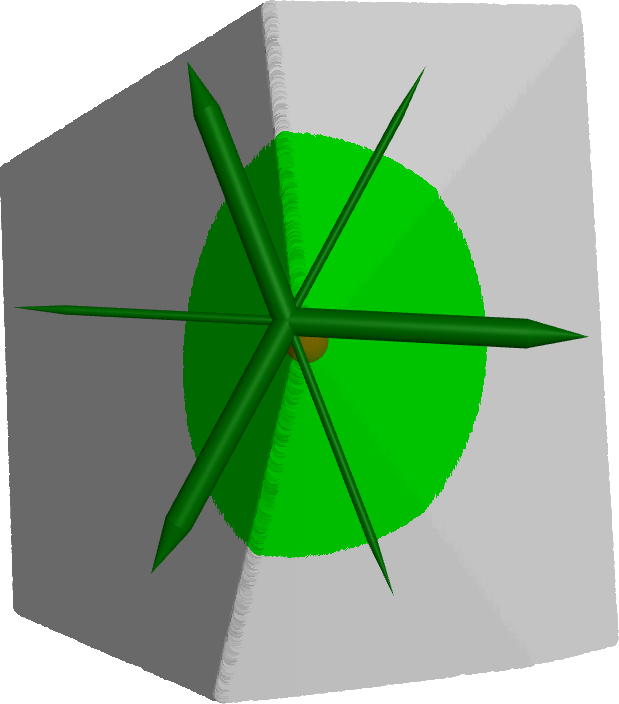}}
  \subcaptionbox{Order 4}{\includegraphics[width=0.15\linewidth]{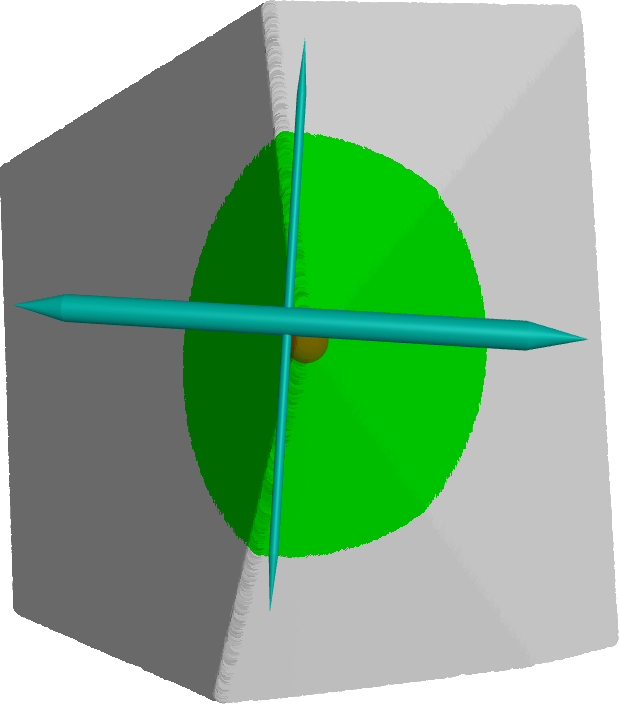}}
  \subcaptionbox{Order 5}{\includegraphics[width=0.15\linewidth]{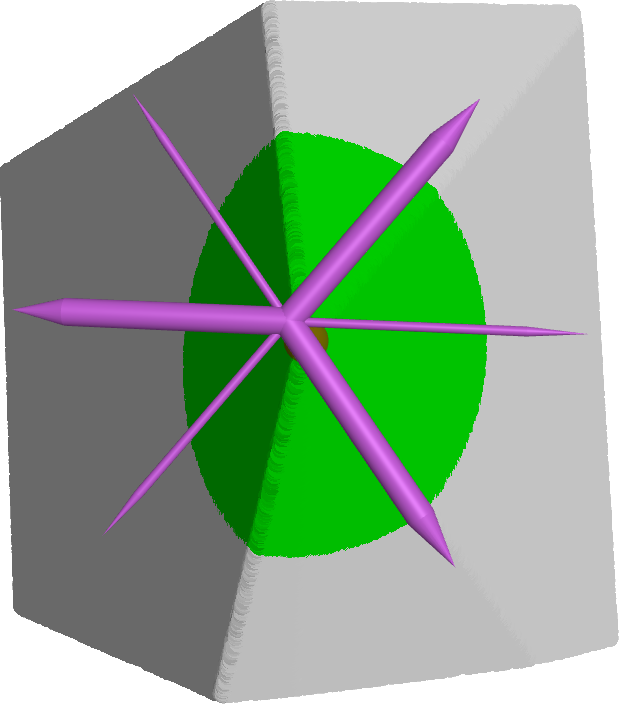}}
  \subcaptionbox{Order 6}{\includegraphics[width=0.15\linewidth]{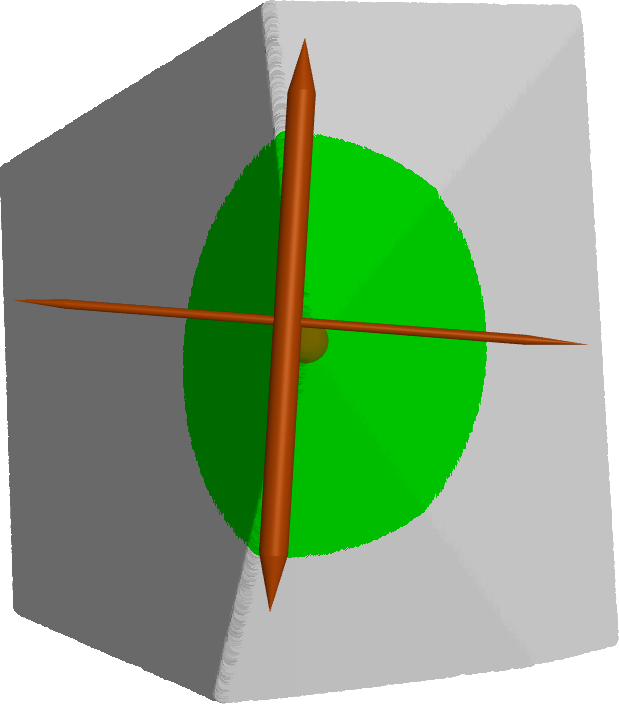}}
  \subcaptionbox{Order 7}{\includegraphics[width=0.15\linewidth]{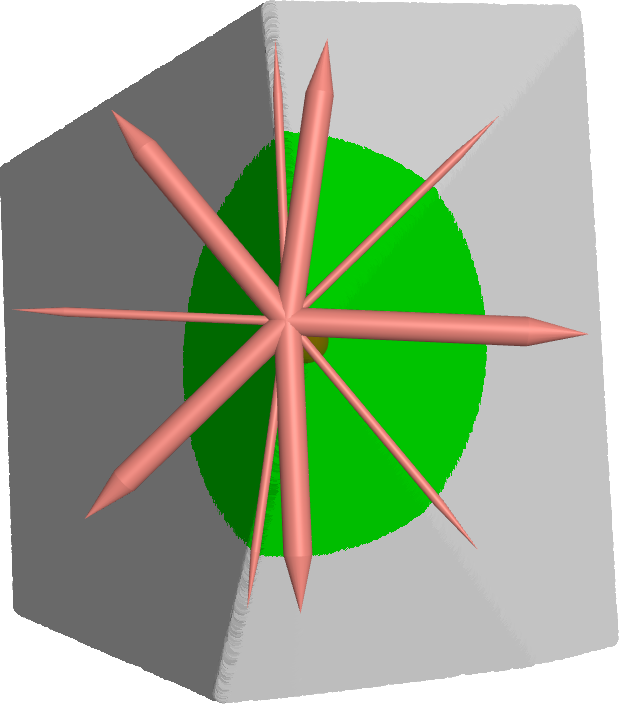}}
 \end{center}

 \caption{Five intersecting planes, the intersection distribution being irregular. Order 5 itself fails to capture this irregularity fully, but the proper intersection directions can be found among orders 5 and 7 directions. \label{fig:inter5plane_irreg}}
\end{figure}

On Figure \ref{fig:cube23}, we show the behavior of order 2 and 3 principal directions along the edges and corners of a synthetic cube. The length of the directions reflects the amplitude of the extremum. Order 3 accounting for some antisymmetry vanishes for edge points (which are symmetric) and order 2 vanishes on the corners (which are antisymmetric).

\begin{figure}[h!]
 \begin{center}
  \includegraphics[width=0.3\linewidth]{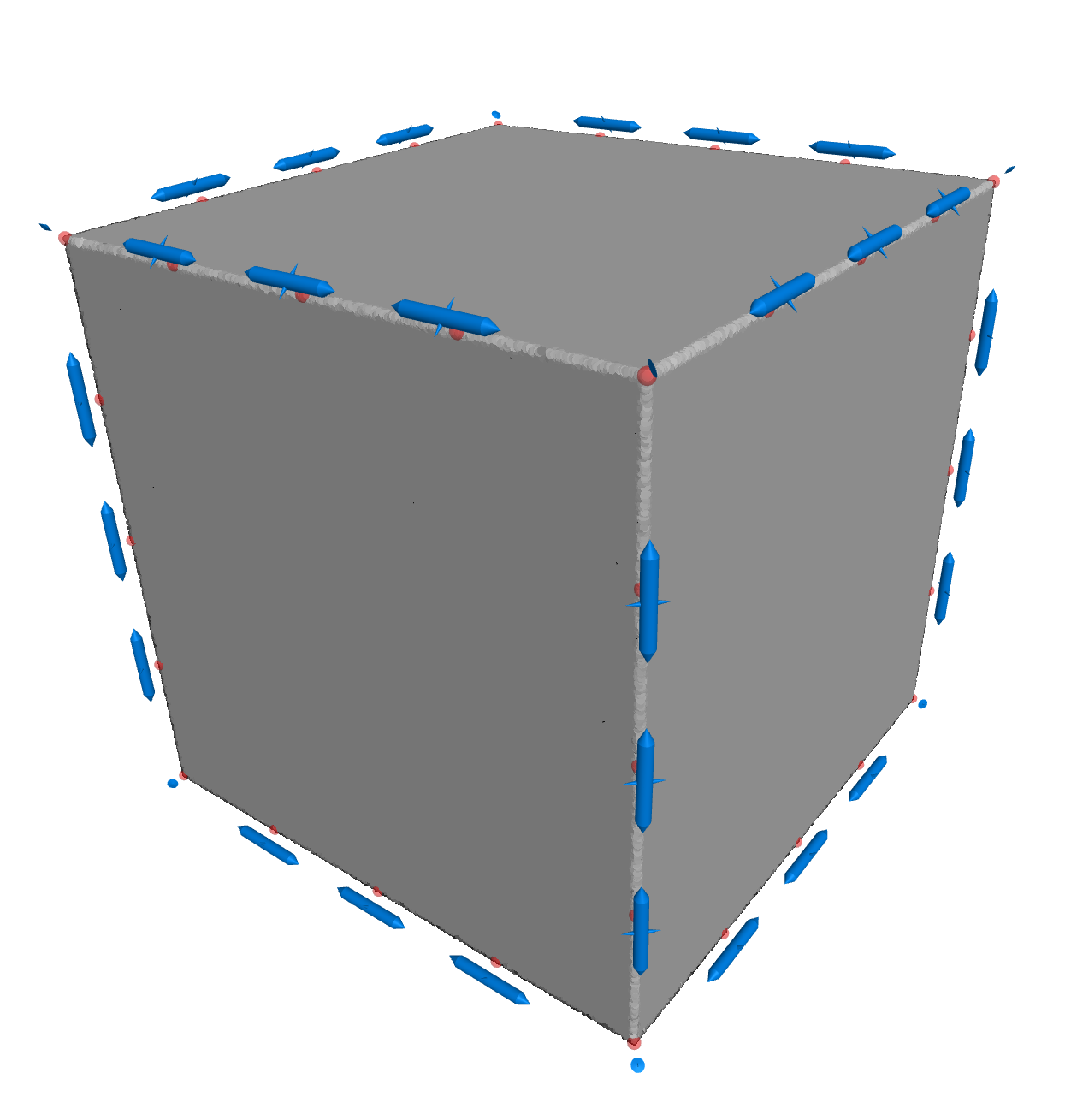}
  \includegraphics[width=0.3\linewidth]{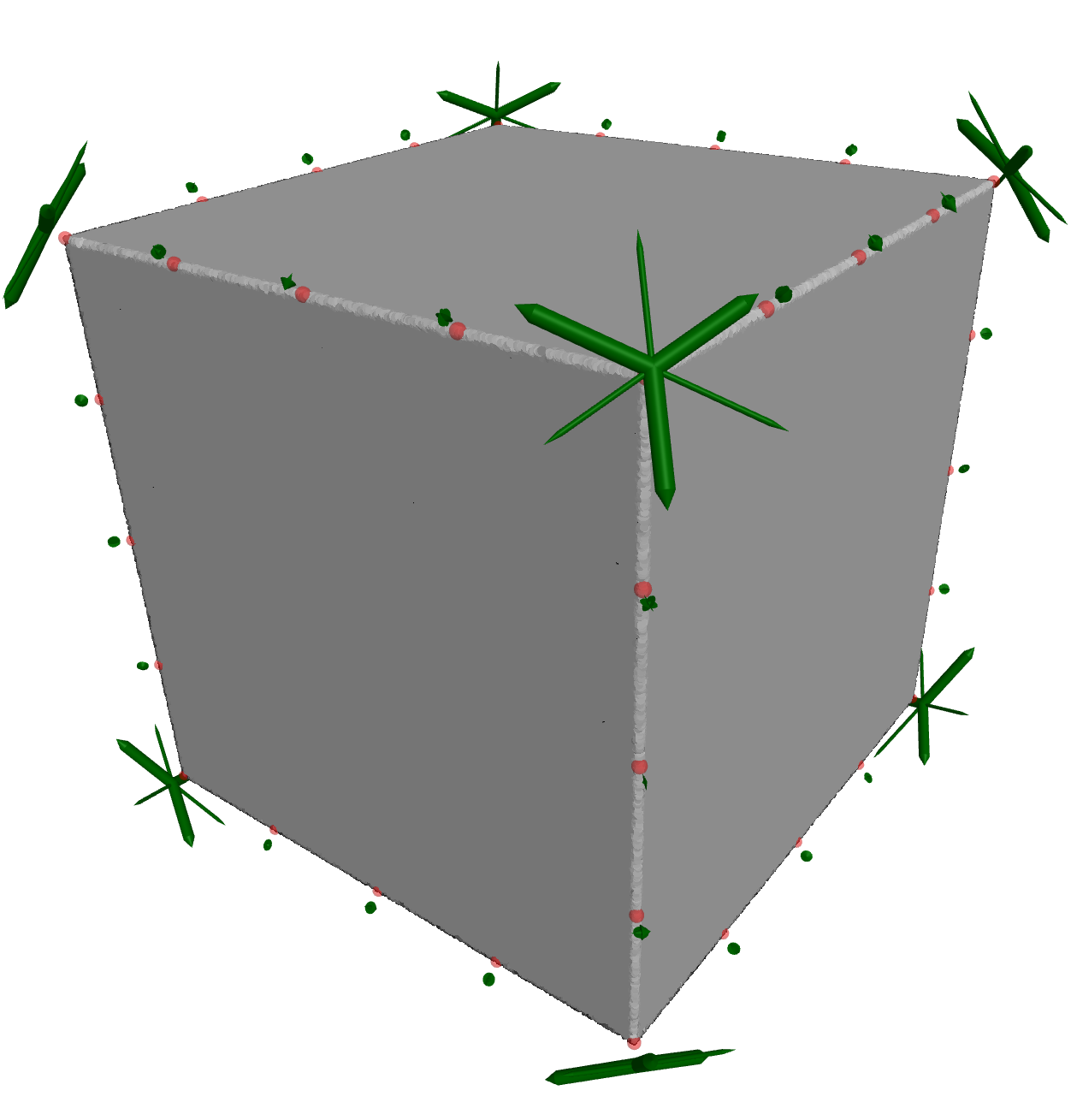}
 \end{center}
\caption{Order 2 and 3 principal directions on the edges and corners of a cube: Order 2 vanishes at the corner points, while order 3 vanishes on the edges of the cube.}
\label{fig:cube23}
\end{figure}

\paragraph*{Real world models}

\begin{figure}[h!]
 \begin{center}
  \includegraphics[width=0.35\linewidth]{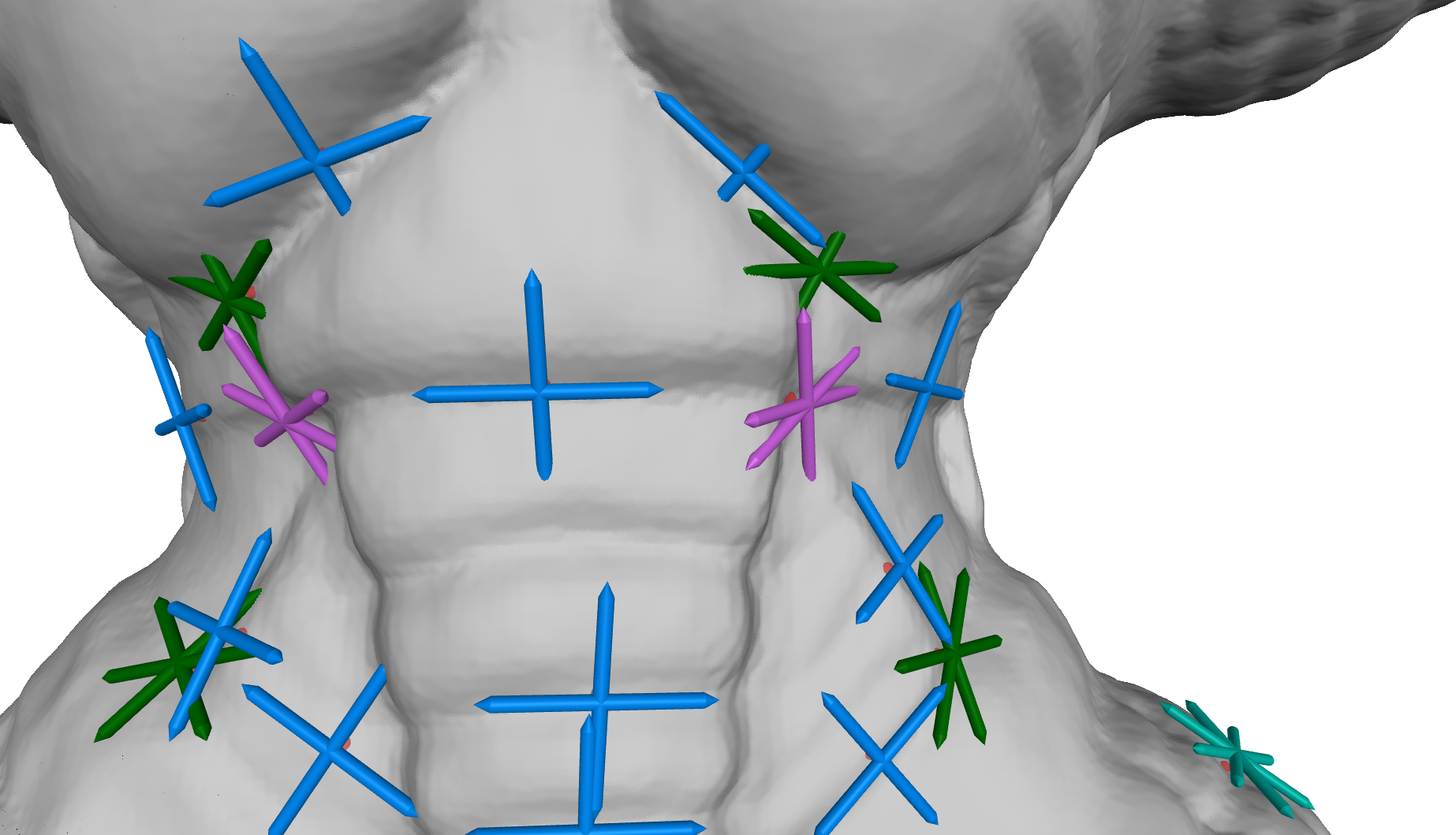}
  \includegraphics[width=0.25\linewidth]{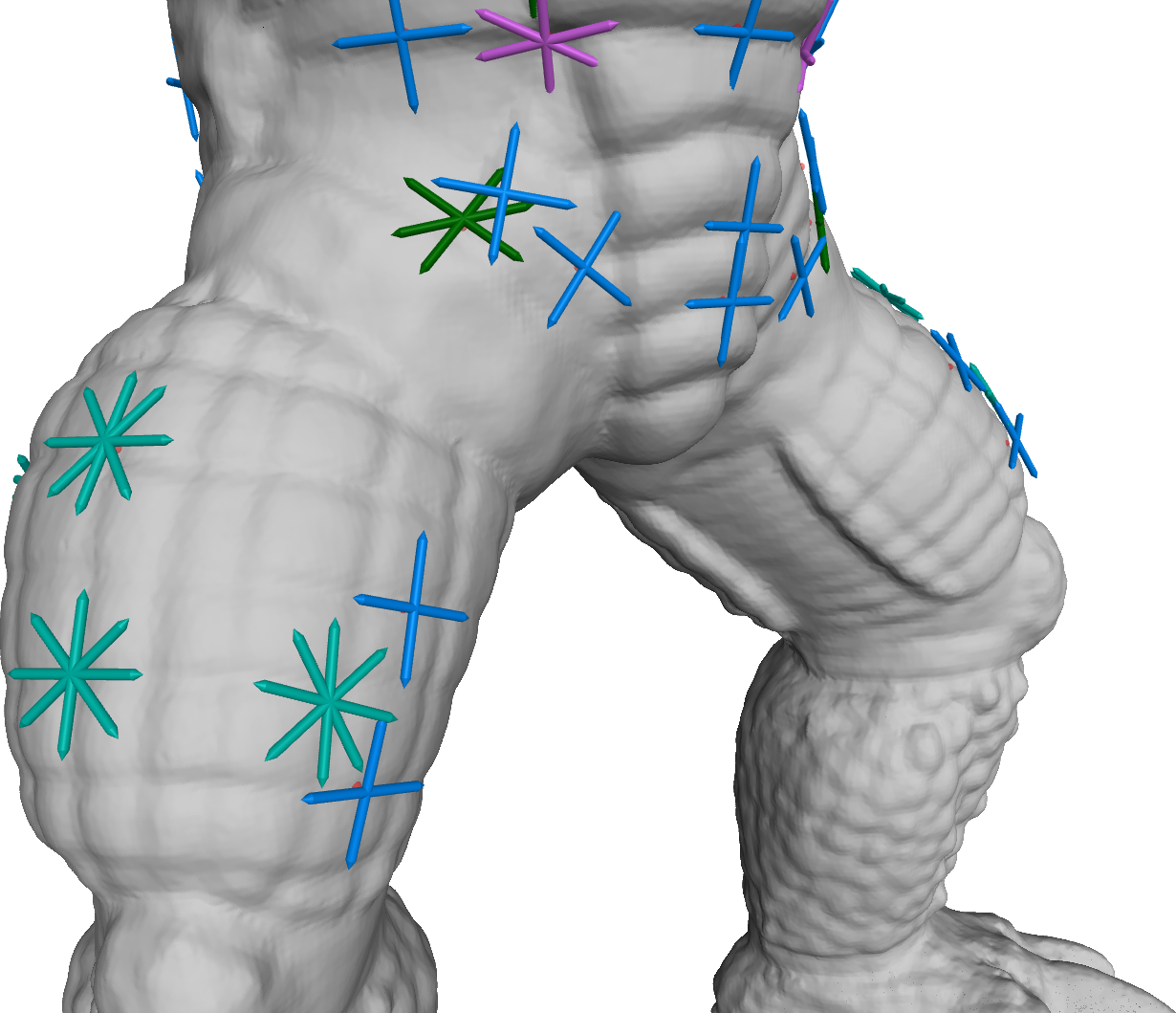}
 \end{center}
\caption{Principal directions of various orders on the torso and leg of the Armadillo (see also Fig.\ref{fig:teaser}). The orders are chosen manually as the most relevant geometrically (order 2 in blue; order 3 in green; order 4 in cyan). For clarity sake only the maximum directions are shown.}
\label{fig:armadillo_vf}
\end{figure}

Figures \ref{fig:teaser} and \ref{fig:armadillo_vf} show some of the principal directions computed by our method on the Armadillo model sampled with $5M$ points. The principal directions orders are chosen manually according to local geometric features. As expected order $2$ accounts well for valleys, order $4$ for valley crossings, and order $3$ for some antisymmetry and monkey-saddle like features.

\begin{figure}[h!]
 \begin{center}
  \includegraphics[width=0.3\linewidth]{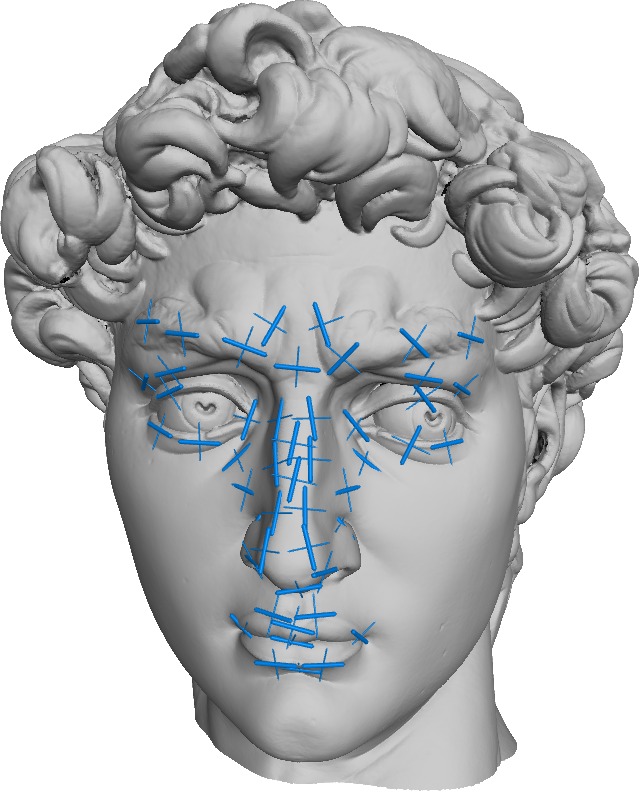}
  \includegraphics[width=0.3\linewidth]{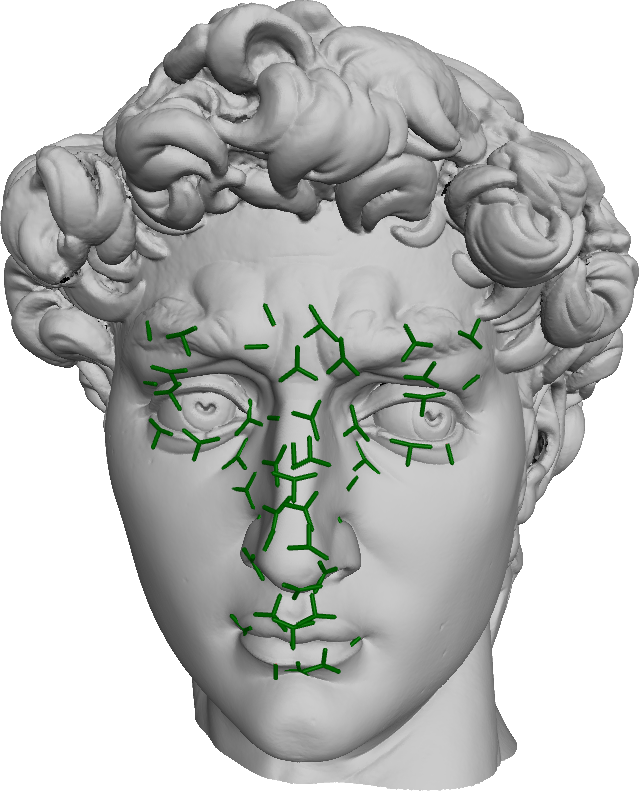}
 \end{center}
\caption{Vector fields for order 2 and 3 on the David head. Left: order 2 principal directions, Right: order 3 principal directions. For clarity sake only maximum principal directions are shown for order 3.}
\label{fig:david}
\end{figure}

Figure \ref{fig:david} shows the principal directions of orders 2 and 3 computed at various locations. While it is obvious that sometimes order 2 is enough (side of the nose), order 3 is meaningful between the eyebrows and around the lips.

Figures \ref{fig:rockerarm} and \ref{fig:fandisk} show some principal directions computed on some more geometric shapes. Here order 3 becomes especially meaningful near corners, at the intersection between 3 smooth surfaces.

\begin{figure}[h!]
 \begin{center}
  \includegraphics[width=0.4\linewidth]{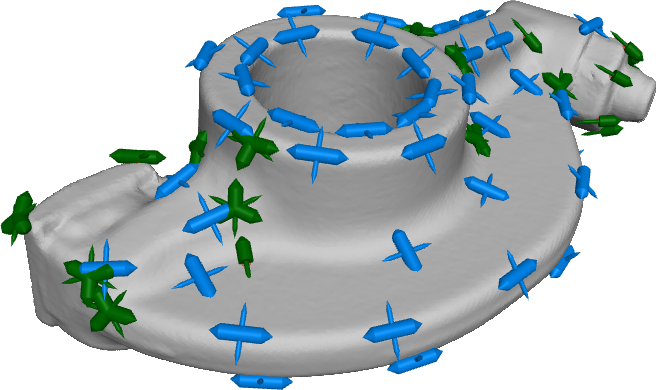}
 \end{center}
\caption{Some principal directions computed on the Rockerarm model, order 2 (blue), and 3 (green).}
\label{fig:rockerarm}
\end{figure}

\begin{figure}[h!]
 \begin{center}
  \includegraphics[width=0.4\linewidth]{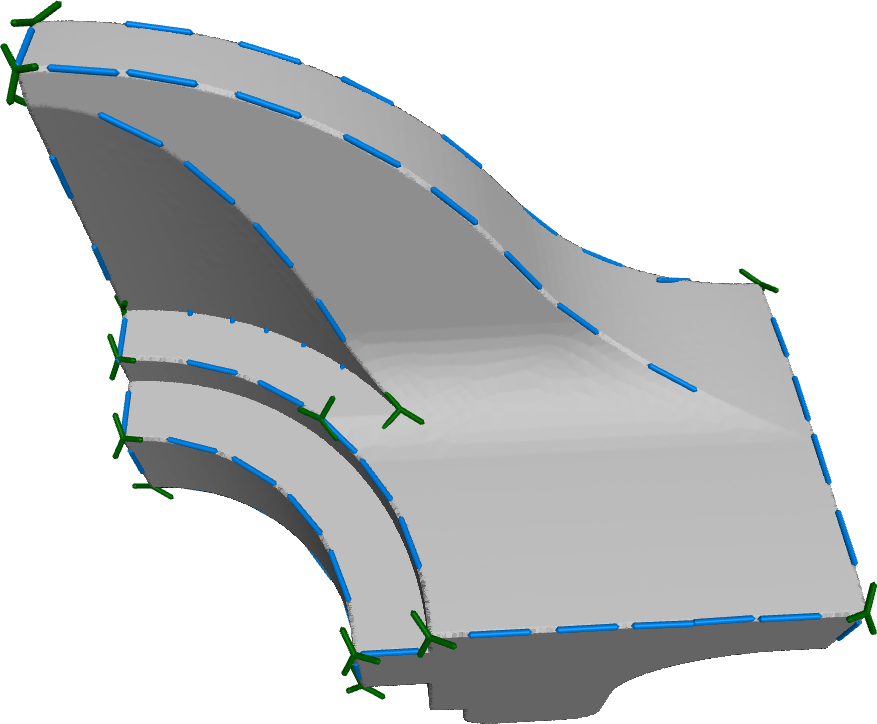}
 \end{center}
\caption{Some principal directions computed on the Fandisk, order 2 (blue), and 3 (green). For clarity sake only maximum directions are shown (see also Fig. \ref{fig:teaser}).}
\label{fig:fandisk}
\end{figure}

\paragraph*{Robustness}

To show the robustness of our principal directions estimation we add some Gaussian noise to the data.  Figure \ref{fig:cube_noisy} shows examples of principal directions estimation of order 2 and 3 on a cube with various noise levels.
 Importantly enough, this robustness does not stem from the principal direction decomposition itself but from the coefficients estimation. Once the coefficients are estimated the principal directions are obtained through function maximum and minimum computations, which can only introduce numerical errors.

\begin{figure}
 \begin{center}
  \includegraphics[width=0.24\linewidth]{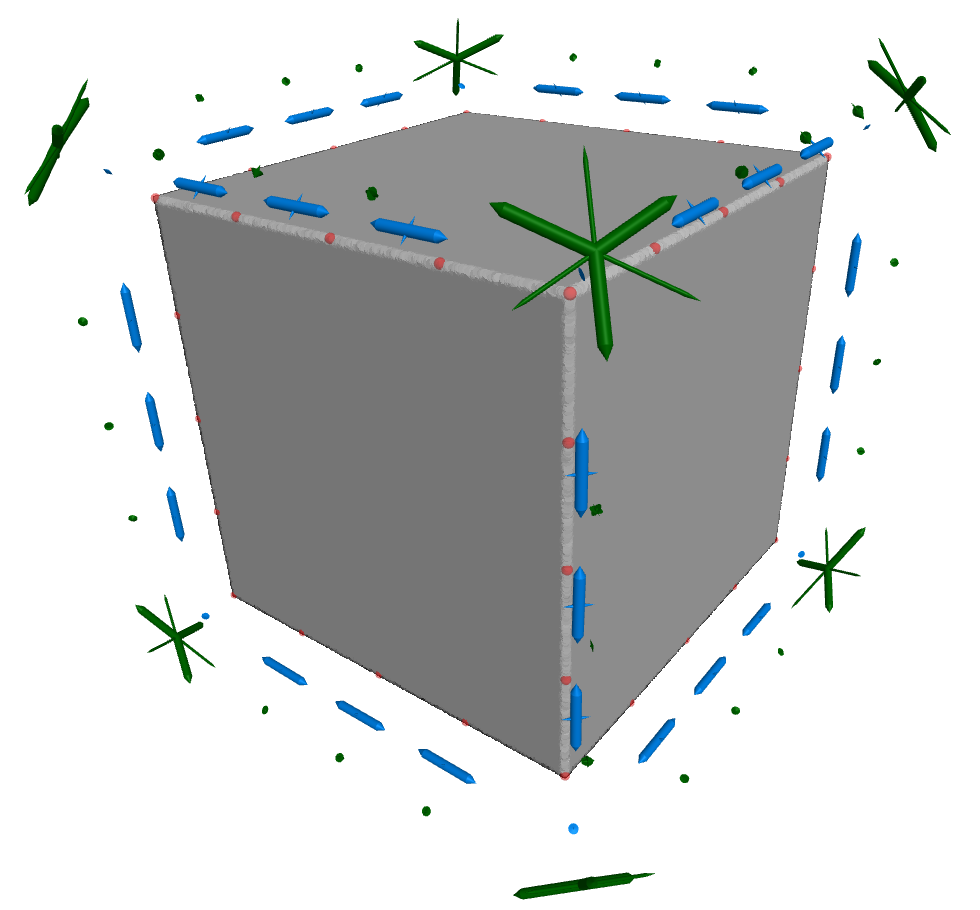}
  \includegraphics[width=0.24\linewidth]{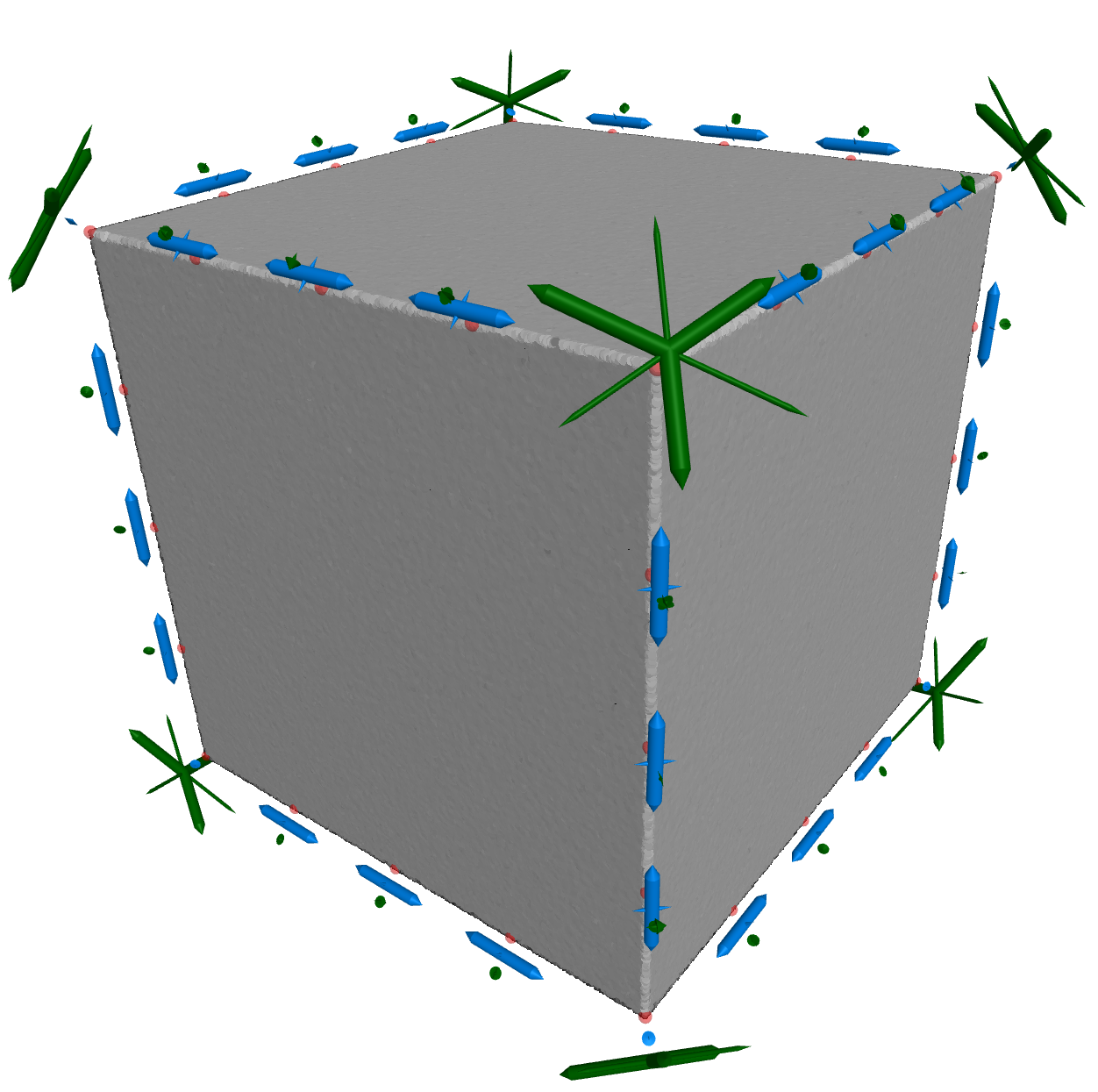}
  \includegraphics[width=0.24\linewidth]{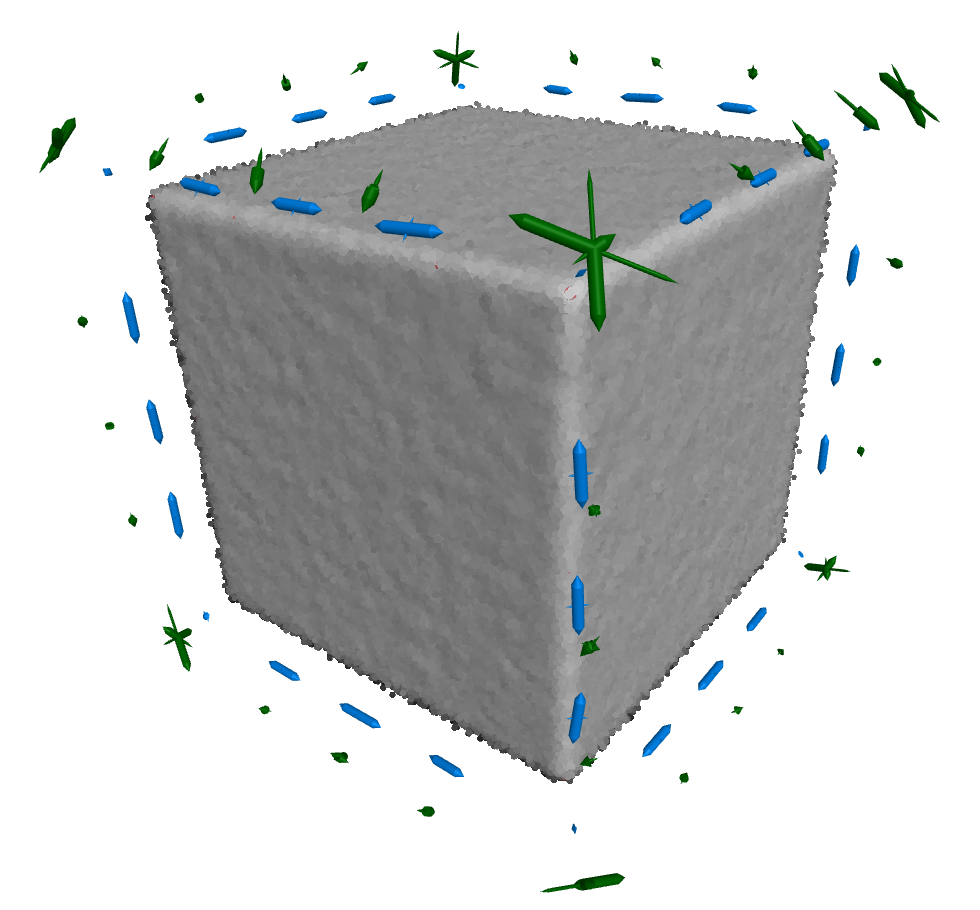}
  \includegraphics[width=0.24\linewidth]{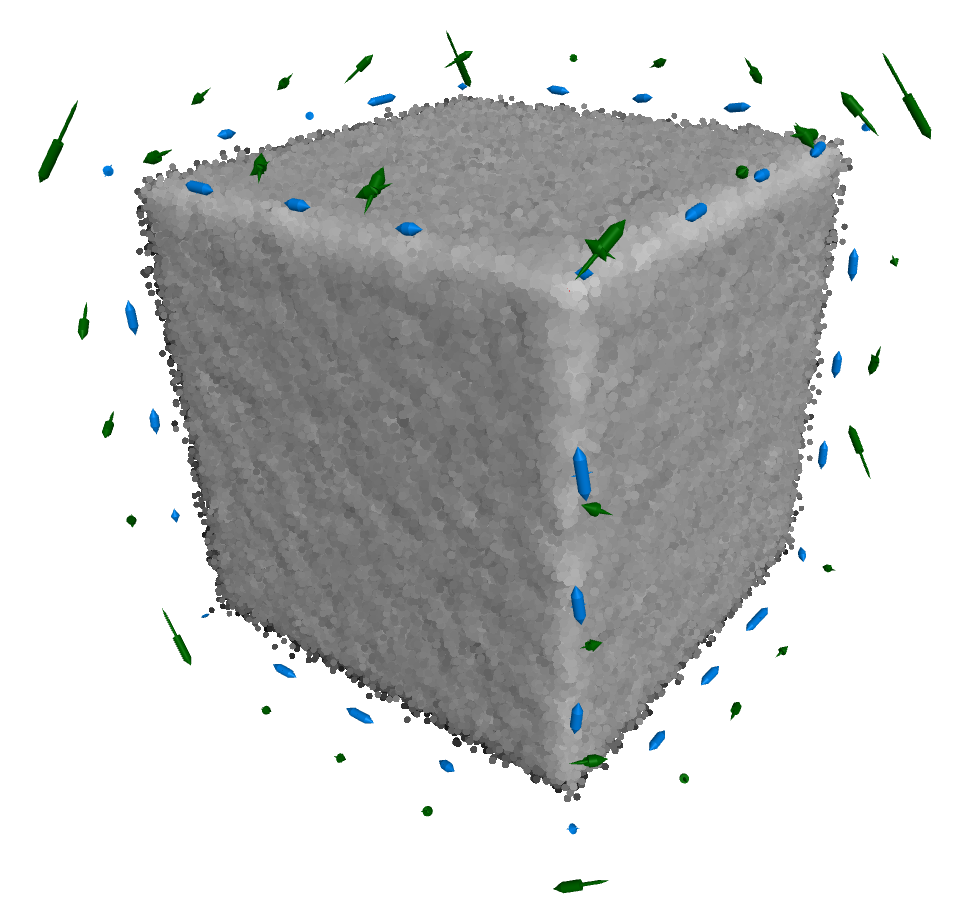}
 \end{center}
 \caption{Principal directions of order 2 and 3 computed on a cube with added Gaussian noise on the positions. From left to right: Noiseless, $\sigma=0.01\%$; $\sigma= 0.05\%$ and $\sigma=0.1\%$ (percentages are given with respect to the shape diagonal)}
\label{fig:cube_noisy}
\end{figure}

\begin{figure}[h!]
 \begin{center}
  \subcaptionbox{$r=50$}{\includegraphics[width=0.15\linewidth]{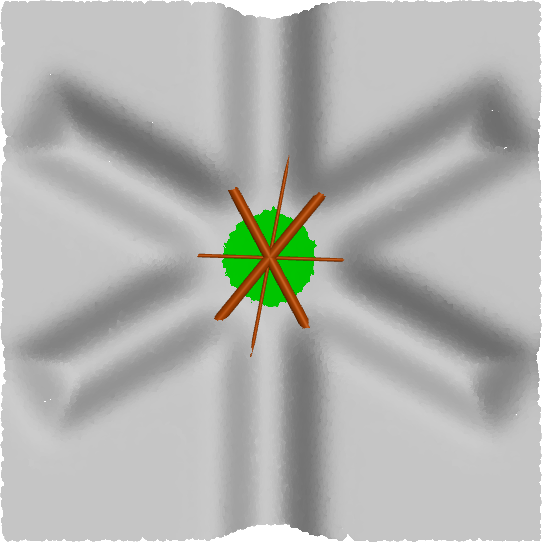}}
  \subcaptionbox{$r=80$}{\includegraphics[width=0.15\linewidth]{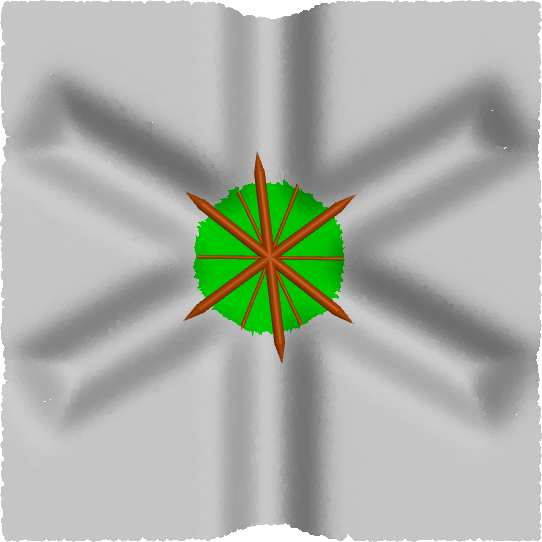}}
  \subcaptionbox{$r=100$}{\includegraphics[width=0.15\linewidth]{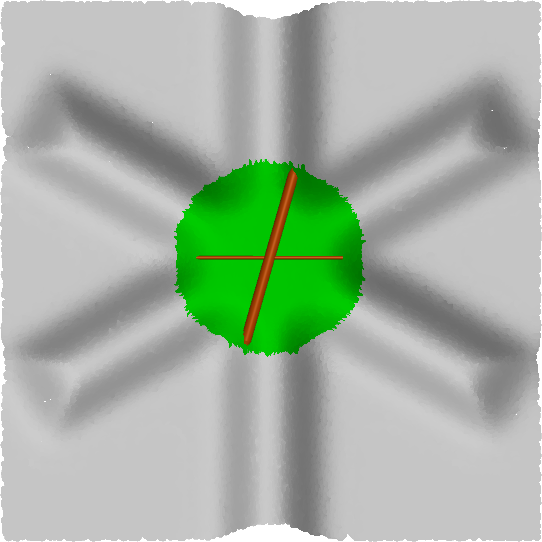}}
  \subcaptionbox{$r=200$}{\includegraphics[width=0.15\linewidth]{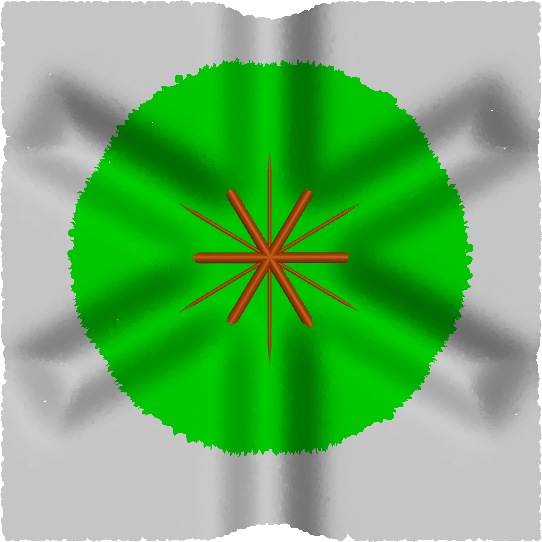}}
 \end{center}
\caption{Detection of order 6 principal directions with increasing radius (the neighborhood is shown in green).}
\label{fig:varying_radius}
\end{figure}

\section{Limitations}

Our definition of principal directions is an extension of the principal curvature directions to higher orders, and as such are continuous on \textit{generic} surfaces. For surfaces that are \textit{umbilical} at every order such as a sphere or a plane, the directions will simply vanish (since no extremum will be found). Further smoothness constraints could be set locally to enforce some consistency, however this goes beyond the scope of this paper.
Our approach also shares a limitation common to many local estimation methods: a radius should be chosen so that the analysis is local enough but also such that the neighborhood encloses enough points. The radius has indeed a direct impact on what is captured by the principal directions as illustrated on Figure \ref{fig:varying_radius}.
Importantly enough, our method does not perform better for curvature principal directions estimation than Osculating Jets~\cite{Cazals05} or APSS~\cite{Guennebaud07}. Our contribution lies rather on the extension to arbitrary orders than on the estimation accuracy itself.
Finally, while it is appealing to consider that order-$3$ principal directions capture three ridges meeting at a single point, some precautions ought to be taken: if the ridges meet at a perfect T-junction, the maximum (or minimum) principal directions will not capture the 3 ridges directions because two maximum order-$3$ principal directions cannot be opposite. Figure \ref{fig:Tjunction_order3} illustrates this behavior (see also Figure \ref{fig:inter5plane_irreg}).

\begin{figure}[h!]
 \begin{center}
  \includegraphics[width=0.4\linewidth]{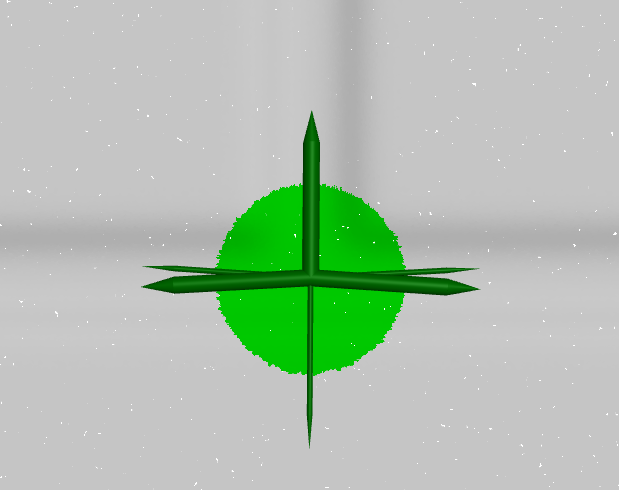}
 \end{center}
 \caption{Order 3 directions computed on a perfect T-junction. The directions can mathematically not take a perfect T shape.}
 \label{fig:Tjunction_order3}
\end{figure}

\section{Conclusion}

In this paper we introduced an extension of principal curvature directions to arbitrary differential orders and showed the link with the eigenvectors and eigenvalues of differential tensors. We showed that these new intrinsic direction fields are relevant on several shapes and can be computed efficiently even with sharp features. As a future work, some global smoothness constraints could be added, to enforce some surrogate direction computations on surfaces where the directions vanish. Many more applications of this new type of principal directions remain to be explored.


\bibliographystyle{siamplain}
\bibliography{biblio}

\begin{thebibliography}{10}

\bibitem{Azencot15}
{\sc O.~Azencot, O.~Vantzos, M.~Wardetzky, M.~Rumpf, and M.~Ben-Chen}, {\em
  Functional thin films on surfaces}, SCA '15, Association for Computing
  Machinery, 2015, pp.~137 -- 146.

\bibitem{Bearzi18}
{\sc Y.~B\'earzi, J.~Digne, and R.~Chaine}, {\em Wavejets: A local frequency
  framework for shape details amplification}, Computer Graphics Forum, 37
  (2018), pp.~13--24.

\bibitem{Brandt17}
{\sc C.~Brandt, L.~Scandolo, E.~Eisemann, and K.~Hildebrandt}, {\em Spectral
  processing of tangential vector fields}, Computer Graphics Forum, 36 (2017),
  pp.~338--353.

\bibitem{Brandt18}
{\sc C.~Brandt, L.~Scandolo, E.~Eisemann, and K.~Hildebrandt}, {\em Modeling
  n-symmetry vector fields using higher-order energies}, ACM Trans. Graph., 37
  (2018).

\bibitem{Cazals05}
{\sc F.~Cazals and M.~Pouget}, {\em Estimating differential quantities using
  polynomial fitting of osculating jets}, Computer Aided Geometric Design, 22
  (2005), pp.~121 -- 146.

\bibitem{Cohen-Steiner03}
{\sc D.~Cohen-Steiner and J.-M. Morvan}, {\em Restricted delaunay
  triangulations and normal cycle}, in Proc. SCG '03, 2003.

\bibitem{Crane10}
{\sc K.~Crane, M.~Desbrun, and P.~Schröder}, {\em Trivial connections on
  discrete surfaces}, Computer Graphics Forum, 29 (2010), pp.~1525--1533.

\bibitem{Diamanti14}
{\sc O.~Diamanti, A.~Vaxman, D.~Panozzo, and O.~Sorkine-Hornung}, {\em
  Designing n-polyvector fields with complex polynomials}, Computer Graphics
  Forum, 33 (2014).

\bibitem{Diamanti15}
{\sc O.~Diamanti, A.~Vaxman, D.~Panozzo, and O.~Sorkine-Hornung}, {\em
  Integrable polyvector fields}, ACM Trans. Graph., 34 (2015).

\bibitem{Guennebaud07}
{\sc G.~Guennebaud and M.~Gross}, {\em Algebraic point set surfaces}, ACM
  Trans. Graph., 26 (2007).

\bibitem{Jakob15}
{\sc W.~Jakob, M.~Tarini, D.~Panozzo, and O.~Sorkine-Hornung}, {\em Instant
  field-aligned meshes}, ACM Trans. Graph., 34 (2015).

\bibitem{Kalogerakis09}
{\sc E.~Kalogerakis, D.~Nowrouzezahrain, P.~Simari, and K.~Singh}, {\em
  Extracting lines of curvature from noisy point cloud}, Computer-Aided Design,
  41 (2009), pp.~282 -- 292.
\newblock Point-based Computational Techniques.

\bibitem{Kalogerakis07}
{\sc E.~Kalogerakis, P.~Simari, D.~Nowrouzezahrai, and K.~Singh}, {\em Robust
  statistical estimation of curvature on discretized surfaces}, in Symposium on
  Geometry Processing, 2007.

\bibitem{Knoppel13}
{\sc F.~Kn\"{o}ppel, K.~Crane, U.~Pinkall, and P.~Schr\"{o}der}, {\em Globally
  optimal direction fields}, ACM Trans. Graph., 32 (2013).

\bibitem{Knoppel15}
{\sc F.~Kn\"{o}ppel, K.~Crane, U.~Pinkall, and P.~Schr\"{o}der}, {\em Stripe
  patterns on surfaces}, ACM Trans. Graph., 34 (2015).

\bibitem{Lai10}
{\sc Y.~{Lai}, M.~{Jin}, X.~{Xie}, Y.~{He}, J.~{Palacios}, E.~{Zhang}, S.~{Hu},
  and X.~{Gu}}, {\em Metric-driven rosy field design and remeshing}, IEEE
  Transactions on Visualization and Computer Graphics, 16 (2010), pp.~95--108.

\bibitem{Levin98}
{\sc D.~Levin}, {\em The approximation power of moving least-squares}, Math.
  Comput., 67 (1998), p.~1517–1531.

\bibitem{Levin15}
{\sc D.~Levin}, {\em Between moving least-squares and moving least-l1}, BIT
  Numerical Mathematics,  (2015), pp.~781--796.

\bibitem{Meyer02}
{\sc M.~Meyer, M.~Desbrun, P.~Schröder, and A.~Barr}, {\em Discrete
  differential geometry operators for triangulated 2-manifolds}, in
  International Workshop on Visualization and Mathematics, 2002.

\bibitem{Nasikun20}
{\sc A.~Nasikun, C.~Brandt, and K.~Hildebrandt}, {\em Locally supported
  tangential vector, n-vector, and tensor fields}, Computer Graphics Forum, 39
  (2020), pp.~203--217.

\bibitem{Panozzo12}
{\sc D.~Panozzo, Y.~Lipman, E.~Puppo, and D.~Zorin}, {\em Fields on symmetric
  surfaces}, ACM Trans. Graph., 31 (2012).

\bibitem{Qi05}
{\sc L.~Qi}, {\em Eigenvalues of a real supersymmetric tensor}, Journal of
  Symbolic Computation, 40 (2005), pp.~1302 -- 1324.

\bibitem{Qi06}
{\sc L.~Qi}, {\em Rank and eigenvalues of a supersymmetric tensor, the
  multivariate homogeneous polynomial and the algebraic hypersurface it
  defines}, Journal of Symbolic Computation, 41 (2006), pp.~1309 -- 1327.

\bibitem{Qi07}
{\sc L.~Qi}, {\em Eigenvalues and invariants of tensors}, Journal of
  Mathematical Analysis and Applications, 325 (2007), pp.~1363 -- 1377.

\bibitem{Ray08}
{\sc N.~Ray, B.~Vallet, W.~C. Li, and B.~L\'{e}vy}, {\em N-symmetry direction
  field design}, ACM Trans. Graph., 27 (2008).

\bibitem{Rusinkiewicz04}
{\sc S.~Rusinkiewicz}, {\em Estimating curvatures and their derivatives on
  triangle meshes}, in Proceedings. 2nd International Symposium on 3D Data
  Processing, Visualization and Transmission, 2004. 3DPVT 2004., 2004,
  pp.~486--493.

\bibitem{Sageman-Furnas19}
{\sc A.~O. Sageman-Furnas, A.~Chern, M.~Ben-Chen, and A.~Vaxman}, {\em
  Chebyshev nets from commuting polyvector fields}, ACM Trans. Graph., 38
  (2019).

\bibitem{Vallet08}
{\sc B.~Vallet and B.~L\'{e}vy}, {\em Spectral geometry processing with
  manifold harmonics}, Computer Graphics Forum, 27 (2008), pp.~251--260.

\bibitem{Vaxman16}
{\sc A.~Vaxman, M.~Campen, O.~Diamanti, D.~Panozzo, D.~Bommes, K.~Hildebrandt,
  and M.~Ben-Chen}, {\em Directional field synthesis, design, and processing},
  Computer Graphics Forum, 35 (2016), pp.~545--572.

\bibitem{Wardetzky07}
{\sc M.~Wardetzky, S.~Mathur, F.~K\"{a}lberer, and E.~Grinspun}, {\em Discrete
  laplace operators: No free lunch}, in Symposium on Geometry Processing,
  Eurographics, 2007, p.~33–37.

\bibitem{Wei01}
{\sc L.-Y. Wei and M.~Levoy}, {\em Texture synthesis over arbitrary manifold
  surfaces}, SIGGRAPH '01, Association for Computing Machinery, 2001, pp.~355
  -- 360.

\end{thebibliography}

\end{document}